\documentclass[11pt]{article}

\usepackage{amssymb,amsmath,amsthm}
\usepackage{geometry}
\usepackage{color}
\usepackage{varwidth}
\usepackage{subfigure}
\usepackage{graphicx}
\usepackage{tikz}
\usetikzlibrary{decorations.markings}
\usepackage{verbatim}

\tikzset{->-/.style={decoration={
  markings,
  mark=at position #1 with {\arrow{>}}},postaction={decorate}}}

\newcommand{\BC}{\mathbb C}
\newcommand{\BN}{\mathbb N}
\newcommand{\BR}{\mathbb R}
\newcommand{\BZ}{\mathbb Z}

\newcommand{\BE}{\mathbb E}

\newcommand{\cB}{\mathcal B}

\newcommand{\cH}{\mathcal H}
\newcommand{\cM}{\mathcal M}
\newcommand{\cP}{\mathcal{P}}

\newcommand{\e}{{\bf e}}
\newcommand{\bxi}{\boldsymbol{\xi}}
\newcommand{\bdelta}{\boldsymbol \delta}

\newcommand{\bL}{\boldsymbol{L}}

\newcommand{\un}[1]{\underline{#1}}
\newcommand{\ux}{{\un{x}}}
\newcommand{\uy}{{\un{y}}}

\newcommand{\ua}{\un{a}}
\newcommand{\uomega}{\un{\omega}}

\newcommand{\p}{{\partial}}

\newcommand{\f}{\mathfrak{f}}
\newcommand{\g}{\mathfrak{g}}

\newtheorem{thm}{Theorem}
\newtheorem{lemma}{Lemma}
\newtheorem{example}[thm]{Example}
\newtheorem{rem}{Remark}

\newtheorem{defi}{Definition}
\newtheorem{cor}{Corollary}

\DeclareMathOperator{\I}{I}

\title{The Spingroup and its actions in discrete Clifford analysis}
\author{H. De Ridder\footnote{Ghent University, Department of Mathematical Analysis, Building S8, Krijgslaan 281, 9000 Gent, Belgium, fax: 0032 9 264 49 87, phone: 0032 9 264 49 49, email: Hilde.DeRidder@UGent.be}, F. Sommen\footnote{Ghent University, Department of Mathematical Analysis, Building S8, Krijgslaan 281, 9000 Gent, Belgium, fax: 0032 9 264 49 87, phone: 0032 9 264 49 56, email: Franciscus.Sommen@UGent.be}}
\date{}

\begin{document}
\maketitle


\begin{abstract}
Recently, it has been established that the discrete star Laplace and the discrete Dirac operator, i.e. the discrete versions of their continuous counterparts when working on the standard grid, are rotation-invariant. This was done starting from the Lie algebra $\mathfrak{so}(m,\mathbb{C})$ corresponding to the special orthogonal Lie group $\textup{SO}(m)$; considering its representation in the discrete Clifford algebra setting and proving that these operators are symmetries of the Dirac and Laplace operators. This set-up showed in an abstract way that representation-theoretically the discrete setting mirrors the Euclidean Clifford analysis setting. However from a practical point of view, the group-action remains indispensable for actual calculations. In this paper, we define the discrete Spingroup, which is a double cover of $\textup{SO}(m)$, and consider its actions on discrete functions. We show that this group-action makes the spaces $\mathcal{H}_k$ and $\mathcal{M}_k$ into $\textup{Spin}(m)$-representations. We will often consider the compliance of our results to the results under the $\mathfrak{so}(m,\mathbb{C})$-action. 
\end{abstract}



\section{Introduction}
From an application point of view, one has always been interested in discrete complex analysis and, more recently, in higher-dimensional function theories both generalizing discrete complex analysis and refining discrete harmonic analysis. This interest has been even further sparked by the increase in computational power and the potential of quickly applying even higher-dimensional function-theoretical results. Pioneering work on discrete holomorphic functions on a complex grid was done in \cite{Ferrand,Isaacs} and research on these discrete holomorphic functions on (more general grids) was continued on in amongst others \cite{Kenyon,Mercat}. When considering a discrete version of Euclidean Clifford analysis (see for example \cite{rood,groen,GM}), foundations were laid in \cite{Faustinothesis,GH,FKS}, although these works often differ in terms of the chosen discrete Dirac operator and/or on the chosen graph on which functions are defined. In this paper, we will restrict ourself to the `split' discrete Clifford algebra; a basic framework established in \cite{FKS,frame} that uses both forward and backward differences. 

\medskip The key notion of discrete Clifford analysis is a discrete Dirac operator, factorising the discrete Laplace operator, leading to a refinement of harmonic analysis. The (massless) Dirac operator finds its origin in particle physics, from the study of elementary particles with spin number one half \cite{D,Weinberg}. It is well known that both the continuous Laplace and Dirac operator are rotation invariant operators, i.e. invariant under the groups SO$(m)$ and Spin$(m)$ respectively, or equivalently, their mutual Lie algebra $\mathfrak{so}(m)$. The space of $\mathbb{C}$-valued harmonic polynomials homogeneous of degree $k$ is in fact a model for an irreducible $\text{SO}(m,\mathbb{C})$-representation with highest weight $(k, 0, \ldots, 0$) \cite{GM, CSVL}. A similar result is true for spinor-valued monogenic polynomials, homogeneous of degree $k$, where the highest weight of the irreducible representation is given by $\left( k + \frac{1}{2}, \frac{1}{2}, \ldots, \frac{1}{2} \right)$ in the case of an odd dimension. Since the space of Dirac spinors $\mathbb{S}$ decomposes as a direct sum of positive and negative Weyl spinors $\mathbb{S}^+ \oplus \mathbb{S}^-$ in even dimension, the space of spinor-valued monogenic polynomials homogeneous of degree $k$ decomposes in even dimension in a sum of exactly two irreducible $\text{SO}(m,\mathbb{C})$-representations with highest weights $\left( k + \frac{1}{2}, \frac{1}{2}, \ldots, \frac{1}{2} \right)$ and $\left( k + \frac{1}{2}, \frac{1}{2}, \ldots, \frac{1}{2},-\frac{1}{2} \right)$. 

\medskip Very recently, the representation-theoretical aspects underlying the discrete counterpart of this function theory, including the rotational invariance of the star-Laplacian and discrete Dirac operator, have been studied. It has been established in recent papers \cite{Rotations,Translations,Hkdecomp,Mkdecomp} that the spaces $\cH_k$ and $\cM_k$ of discrete harmonic, respectively discrete monogenic $k$-homogeneous polynomials are invariant under the action of the special orthogonal Lie algebra $\mathfrak{so}(m,\mathbb{C})$. However, up till now we were always restricted to the use of the Lie algebra $\mathfrak{so}(m,\BC)$ as the (action of the) the special orthogonal Lie group $\textup{SO}(m)$ (or its double cover the Spingroup) was not yet defined. 
In this paper, the aim is to do just that, define and consider a discrete Spingroup which is a double cover of the special orthogonal Lie group. We considered the Spingroups action on spaces of discrete harmonic resp. monogenic polynomials. Although it may abstractly be seen as `just' another realisation of the Spingroup, it is novel as the definition of the Spingroup does not use vectors in the discrete vector variables, as one would expect, but vectors in some recently defined (see \cite{Translations}) operators $R_j$. The fact that there is a discrete Spingroup with similar actions as in Euclidean Clifford analysis makes it clear that, although we are restricted to the points to the grid, rotations are also inherently present in the discrete Clifford analysis setting.

\medskip In Section \ref{sec:preliminaries}, we give a short overview of the necessary definitions and operators of discrete Clifford analysis. In Section \ref{sec:Spingroup} we introduce the definitions of discrete Spingroups and show that they are double covers of the special orthogonal group $\textup{SO}(m)$. In Section \ref{sec:repr} we define several Spingroup actions on the space of discrete polynomials and, extending by means of the Taylor series, the space of all discrete functions. We conclude this section by making the connection to the corresponding Lie algebra. In Section \ref{sec:twodims} we consider the first non-trivial example, i.e. the two-dimensional case; we give explicit examples and compare to earlier results. In Section \ref{sec:distr} we extend our Spingroup action to discrete distributions and considered some basic examples in two dimensions. Finally, in Sections \ref{sec:irrep} and \ref{sec:irrep2}, we consider irreducible representations of integer and half-integer highest weights by constructing the corresponding highest weight functions.

\section{Preliminaries}\label{sec:preliminaries}
Let $\BR^m$ be the $m$-dimensional Euclidean space with orthonormal basis $e_j$, $j=1,\dots,m$ and consider the Clifford algebra $\BR_{m,0}$ over $\BR^m$, i.e. the multiplication of two basis elements must satisfy the anti-commutator rule $e_i \, e_j + e_j \, e_i = 2\,\delta_{ij}$.
Passing to the so-called `split' discrete Clifford setting, see e.g. \cite{frame,SW}, we embed the Clifford algebra $\BR_{m,0}$ into the bigger complex one $\BC_{2m,0}$ and introduce forward and backward basis elements $ \textbf{e}_j^{\pm}$ by splitting the basis elements $e_j = \e_j^+ + \e_j^-$ in forward and backward basis elements $\e_j^\pm$ which sum up to the original basis elements. These $\e_j^\pm$ satisfy the following anti-commutator rules
$$
\left\{ \e_j^+, \e_k^+ \right\} = \left\{ \e_j^-, \e_k^- \right\} =0, \qquad \left\{ \e_j^+, \e_k^- \right\} = \delta_{jk}, \qquad j, k = 1,\ldots, m,
$$
which follow from the principles of dimensional equivalence and reflection invariance \cite{frame}. We denote furthermore $e_j^\perp = \e_j^+ - \e_j^-$, then $e_j^\perp e_j = \e_j^+ \e_j^- - \e_j^- \e_j^+$.
 
Now consider the standard equidistant lattice $\mathbb{Z}^m$. The partial derivatives $\p_{x_j}$ used in Euclidean Clifford analysis (see e.g. \cite{rood,groen}) are replaced by forward and backward differences $\Delta_j^{\pm}$, $j=1,\dots, m$, acting on discrete Clifford-valued functions $f$ as follows:
$$
\Delta^+_j[f](x) = f(x + e_j) - f(x), \qquad \Delta^-_j[f](x) = f(x) - f(x- e_j), \qquad x \in \BZ^m.
$$
An appropriate definition of a discrete Dirac operator $\p$ factorizing the discrete Laplace operator $\Delta$, i.e. satisfying $\p^2 = \Delta$, is obtained by combining the forward and backward basis elements with the corresponding forward and backward differences, more precisely  
$$
\p = \sum_{j=1}^m \left(\e^+_j \Delta^+_j + \e^-_j \Delta^-_j \right) = \sum_{j=1}^m \p_j.
$$
The discrete Dirac operator is complemented with a vector variable operator $\xi$ of the form $\xi = \sum_{j=1}^m \left( \e^+_j X^-_j + \e^-_j X^+_j \right) = \sum_{j=1}^m \xi_j$ and a discrete Euler operator $\BE$ to generate an $\mathfrak{osp}(1|2)$-realisation, cf. \cite{SW}. This means that they satisfy the usual intertwining relations 
$$
\left\{Ê\p, \xi \right\} = 2\,\BE + m, \qquad \left[ \p, \BE \right] = \p, \qquad \left[ \xi, \BE\right] = -\xi. 
$$
On the co-ordinate level, this is expressed by means of the relations $\p_j \, \xi_j - \xi_j \p_j = 1$ and
$$
\left\{ \p_j, \xi_k \right\} = \left\{ \xi_j, \xi_k \right\} = \left\{ \p_j, \p_k \right\} = 0, \qquad j \neq k.
$$
\begin{defi}
A discrete (Clifford-algebra valued) function is discrete harmonic (resp. (left) discrete monogenic) in a domain $\Omega \subset \mathbb{Z}^m$ if $\Delta f(x) = 0$ (resp. $\p f(x) = 0$), for all $\ux \in \Omega$. 
\end{defi}
The space of all discrete Clifford-algebra valued harmonic (resp. monogenic) polynomials is denoted $\cH$ (resp. $\cM$) while the space of discrete Clifford-algebra valued harmonic (resp. monogenic) homogeneous polynomials of degree $k$ is denoted $\cH_k$ (resp. $\cM_k$).

\medskip The natural powers $\xi^k_j[1]$ of the operator $\xi_j$ acting on the ground state 1 are the basic discrete homogeneous polynomials of degree $k$ in the variable $x_j$, replacing the basic powers $x_j^k$ in the continuous setting and constituting a basis for all discrete polynomials, cf. \cite{CK}. The skew-Weyl relations imply that $\p_\ell \, \xi_j^k[1] = \delta_{j,\ell}\, k \,\xi_j^{k-1}[1]$. An explicit formula for the polynomials $\xi^k[1]$ is given in \cite{SW}.
An important property of these polynomials is the fact $\xi_j^k[1](x_j)=0$ for $k \geqslant 2\,|x_j|+1$ which implies the absolute convergence of the Taylor series of any discrete function. Every discrete function, defined on $\mathbb{Z}^m$, can be expressed in terms of these basis discrete homogeneous polynomials by means of its Taylor series expansion around the origin, cf. \cite{Taylor}. 

\medskip In \cite{Translations}, we defined the mutually anti-commuting vector-valued operators $R_j$, satisfying
$$
R_j[1] = e_j, \qquad \left\{ R_j, \xi_k \right\} = 2\,R_j \, \xi_j \, \delta_{j,k}, \qquad \left\{ R_j, \p_k \right\} = 2\,R_j \, \p_j \, \delta_{j,k}.
$$
Note that in combination with the operators $\xi_j$ and $\p_j$, we obtain mutually commuting operators $\xi_j R_j$ and $\p_j R_j$, $j=1,\ldots,m$, for which one can easily check that they also generate an $\mathfrak{osp}(1|2)$-realisation: 
\begin{align*}
\left[ \xi_j R_j, \xi_k R_k \right]Ê&= \left[ \xi_j R_j, \p_k R_k \right]Ê= 0, \\
\left[ \p_j R_j, \xi_k R_k \right]Ê&= \delta_{j,k}. 
\end{align*}
These operators allowed us to define $\mathfrak{so}(m,\mathbb{C})$-generators $L_{a,b}$ resp. $dR(e_{a,b})$ within the discrete Clifford setting, which are symmetries of the discrete Laplace operator $\Delta$ resp. discrete Dirac operator $\p$.
\begin{defi}
For $a \neq b$, we define 
\begin{align*}
L_{a,b} &= R_b\,R_a \left( \xi_a\,\p_b + \xi_b \, \p_a\right), \\
dR(e_{a,b}) &= R_b\,R_a \left( \xi_a\,\p_b + \xi_b \, \p_a - \frac{1}{2}Ê\right) = L_{a,b} - \frac{1}{2}Ê\, R_b\,R_a.
\end{align*}
For $a = b$ let $L_{a,a} = dR(e_{a,a}) = 0$. 
\end{defi}
The spaces of discrete spherical harmonics $\mathcal{H}_k$ of degree $k$ and discrete spherical monogenics $\cM_k$ of degree $k$ are (not irreducible) representations of $\mathfrak{so}(m,\BC)$, their decomposition into irreducible parts was recently considered in \cite{Hkdecomp, Mkdecomp}.

\medskip A second set of vector-valued operators $S_j\,e_j^\perp$ was obtained in \cite{Translations}, where $S_j$ now denotes the classical reflection in the $x_j$-direction, which lead to a second set of $\mathfrak{so}(m,\mathbb{C})$-generators $L^\perp_{a,b}$ and $dR^\perp(e_{a,b})$. Similarly, we will find in the next section two separate discrete Spingroups, one involving the operators $R_j$ and the other involving the operators $S_j \, e_j^\perp$.

\section{Discrete Spingroup}\label{sec:Spingroup}
In this section we define a discrete Spingroup $\textup{Spin}(m)$ and show that it is a double cover of the special orthogonal group $\textup{SO}(m)$. The structure of this proof reflects the proof that the Spingroup in Euclidean Clifford analysis is a double cover of $\textup{SO}(m)$, see for example \cite{groen}. However, there are two ways in which both settings are different: first of all the discrete Spingroup $\textup{Spin}(m)$, as defined below, will consist purely of vectors in the operators $R_j$, $j=1,\ldots,m$. Elements of this Spingroup will thus have to act on a discrete function before one can consider the value in a point of the grid. Second, as the operators $R_j$ behave as generators of a Cliffordalgebra of signatuur $(m,0)$, i.e. $R_j^2 = +1$, a lot of steps differ in minus-signs. As will be explained in section \ref{sec:orthspingroup}, there is a second (orthogonal) Spingroup $\textup{Spin}^\perp(m)$ defined within discrete Clifford analysis involving the operators $S_j \, e_j^\perp$ which generate a Clifford algebra of signature $(0,m)$. We choose to omit the proof of that section and just refer to \cite{groen} and instead give the proof involving the operators $R_j$ explicitly.

\medskip Denote with $\BR^1_m$ the linear vectorspace $\BR^1_m = \left\{ \uomega = \sum_{j=1}^m \omega_j \, R_j \, : \, \omega_j \in \BR, \, j=1,\ldots,m \right\}$. Even though $\uomega$ is in fact an operator, we will also call it a vector, which is justifiable since its action of the groundstate $1$ gives us actual vectors: $\uomega[1] = \sum_{j=1}^m \omega_j \, e_j$. Then unit vectors are operators $\uomega \in \BR^1_m$ such that $|\uomega|^2 =: \sum_{j=1}^m \omega_j^2 = 1$. Since $\left\{ R_j, R_k \right\} = 2\,\delta_{j,k}$, this implies that 
$$
\uomega \, \uomega = \sum_{j=1}^m \omega_j^2\, R_j^2 + \sum_{j=1}^m \sum_{k \neq j}Ê\omega_j \, \omega_k \, R_j \, R_k = \sum_{j=1}^m \omega_j^2  = 1. 
$$
For a point $\ua = \left( a_1, \ldots, a_m \right) \in \BZ^m$ we can consider the corresponding vector 
$\ua = \sum_{j=1}^m a_j \, R_j \in \BR^1_m$ which we will also denote $\ua$. 

\medskip Throughout this paper we will consider the following (anti-)involutions on $\BC_{2m,0}$:
\begin{itemize}
\item The reversion $a \mapsto a^\ast$, which is defined on the basis elements $\left( \e_j^\pm\right)^\ast = \e_j^\pm$ and it is linearly extended to the entire Clifford algebra as $(a\,b)^\ast = b^\ast \, a^\ast$. We will also extend their action also to $\BR^1_m$ by $R_j^\ast = R_j$; this is motivated by the fact that $R_j = \e_j^+ \, R_j^+ + \e_j^- \, R_j^-$ with $R_j^\pm $ scalar operators, see \cite{Translations}. 

\item The conjugation $a \mapsto \overline{a}$ is the composition of the complex conjugation with the action, defined on basis elements as $\overline{\e_j^\pm}Ê= -\e_j^\pm$ and linearly extended to the whole Clifford algebra as $\overline{ab} = \overline{b}Ê\, \overline{a}$. We will also consider its action on $\BR^1_m$ where in particular $\overline{R_j} = - R_j$. 

\item The main involution $a \mapsto \widetilde{a}$ which is $\widetilde{a}Ê= \overline{a}^\ast$. In particular it holds that $\widetilde{R_j} = - R_j$.
\end{itemize} 

\begin{defi}
Consider two operators $X$ and $Y$, then we define 
$$
\left\langle X, Y \right\rangle = \frac{1}{2}Ê\left( X\,Y + Y \, X \right).
$$
\end{defi}
For two vectors $\ua$ and $\un{b} \in \BR^1_m$, we find the inner product $\left\langle \ua, \un{b}Ê\right\rangle = \sum_{j=1}^m a_j \, b_j \in \BR$. Two vectors $\ua$ and $\un{b}$ are then called orthogonal if and only if $\left\langle \ua, \un{b}Ê\right\rangle  = 0$. Each vector $\ua \neq 0$ of $\BR^1_m$ is invertible, with inverse element $\ua^{-1} = \frac{\ua}{|\ua|^2}$.

\begin{defi}
The (discrete) Clifford group is the multiplicative group
$$
\Gamma(m) = \left\{ \prod_{i=1}^n \uomega_i \, : \, n \in \BN, \ \uomega_i \in \BR_m^1 \backslash \{0\} \right\}.
$$
\end{defi}
Let $M = \left\{Ê1, \ldots, m\right\}$. Any element $a \in \Gamma(m)$ can be decomposed as 
$$
a = \sum_{A \subseteq M } a_A \, R_A, \qquad a_A \in \BR.
$$
For $A = \left\{ a_1, \ldots, a_k \right\}Ê\subseteq M$ with $1 \leqslant a_1 < a_2 < \ldots < a_k \leqslant m$, we denote $R_A = R_{a_1}Ê\ldots R_{a_k}$ and
$$
|a|^2 = \sum_{A}Ê|a_A|^2.
$$

 \begin{lemma}
For $a, b \in \Gamma(m)$ it holds that 
$$
a \, a^\ast = |a|^2, \qquad \overline{a} \, \widetilde{a} = |a|^2, \qquad |a b| = |a|\,|b|. 
$$
\end{lemma}
\begin{proof}
By definition, an element $a \in \Gamma(m)$ consists of products of non-zero vectors $a = \uomega_1 \uomega_2 \ldots \uomega_k$, with $\uomega_i \in \BR^1_m \backslash \{0\}$. For $\uomega_i$, we find that $\uomega_i \, \uomega_i^\ast = \uomega_i \, \uomega_i = |\uomega_i|^2$ and so we get that
$$
a \, a^\astÊ= \uomega_1 \, \uomega_2 \ldots \uomega_k \, \uomega_k^\astÊ\ldots \uomega_2^\ast \, \uomega_1^\ast = |\uomega_1|^2 \, |\uomega_2|^2 \, \ldots \, |\uomega_k|^2 > 0.
$$
We thus see that $a \, a^\ast$ is scalar. If on the other hand, we decompose $a = \sum_{A}Êa_A\,R_A$, then the (only) scalar part in the product $a \, a^\ast$ is $\sum_{A}Êa_A^2$. We may conclude that 
$$
a \, a^\ast = \sum_{A} a_A^2 = |a|^2. 
$$
Analogously, we can show that $a^\ast \,a  = |a|^2$. We thus see that for $a \in \Gamma(m)$: $a^{-1} = \dfrac{a^\ast}{|a|^2}$. From this it also follows that $(a^{-1})^\ast = \frac{a}{|a|^2} = (a^\ast)^{-1}$. 

\medskip If we now take $a, b \in \Gamma(m)$, then $|ab|^2 = (ab) \, (ab)^\ast = a \, b \, b^\ast \, a^\ast = |a|^2\,|b|^2$ and consequently $|ab| = |a| \, |b|$. Finally, consider 
\begin{align*}
(\widetilde{a})^{-1} &= \left( (-1)^k \, \uomega_1Ê\ldots \uomega_k \right)^{-1} = (-1)^k \, \frac{\uomega_k}{|\uomega_k|^2} \ldots \frac{\uomega_1}{|\uomega_1|^2} = (-1)^k \, \frac{a^\ast}{|a|^2}
= \frac{\overline{a}}{|a|^2}.
\end{align*}
\end{proof}

\medskip With every $a \in \Gamma(m)$, we can introduce the corresponding linear transformation $\chi(a): \BR^1_m \to \BR^1_m$: 
$$
\chi(a)(\ux) = a \, \ux \, (\widetilde{a})^{-1}.
$$

\begin{lemma}
Let $a \in \Gamma(m)$ and $\ux \in \BR^1_m$, then it holds that $a \, \ux \, (\widetilde{a})^{-1} \in \BR^1_m$ and the map 
$$
\chi(a): \BR^1_m \to \BR^1_m: \ \ux \mapsto a \, \ux \, (\widetilde{a})^{-1}
$$
is a bijective isometry, i.e. $\chi(a) \in O(m)$. 
\end{lemma}

\begin{proof}

Take $\ux$, $\uy \in \BR^1_m$, then $\ux \, \uy + \uy \, \ux = 2 \left\langle \ux, \uy \right\rangle \in \BR$. Multiplication on the right of $\uy \, \ux = 2 \,\langle \ux, \uy \rangle - \ux\,\uy$ with $\uy$ shows that $\uy \, \ux\,\uy = 2\,\langle \ux,\uy \rangle \, \uy - |\uy|^2 \, \ux$. This is a linear combination of two elements of $\BR^1_m$ with real coefficients and as such also an element of $\BR^1_m$. Since the inverse element of $\widetilde{\uy} \in \BR^1_m$ is given by $\displaystyle (\widetilde{\uy})^{-1} = \dfrac{-\uy}{|\uy|^2}$, we see immediately that 
$$
\chi(\uy)(\ux) = \uy \, \ux \, \widetilde{\uy}^{-1}Ê= -\frac{\uy\,\ux\,\uy}{|\uy|^2} = - \frac{2\,\langle \ux,\uy \rangle}{|\uy|^2}  \, \uy + \ux \in \BR^1_m.
$$

\medskip Now take $a \in \Gamma(m)$, then $a = \uy_1 \, \uy_2 \ldots \uy_n$ for some $\uy_i \in \BR^1_m \backslash \{0\}$ and $\widetilde{a}Ê= (-1)^n \, \uy_1 \ldots \uy_n$. Thus $(\widetilde{a})^{-1} =  (-1)^n \, \uy_n^{-1}Ê\ldots \uy_1^{-1}$ and hence 
$$
\chi(a)(\ux) = a \, \ux \, (\widetilde{a})^{-1} = (-1)^n \, \uy_1 \, \ldots \, \uy_n \, \ux \, \uy_n^{-1} \, \ldots \, \uy_1^{-1} \in \BR^1_m.
$$

\medskip \noindent The map $\chi(a)$ is clearly injective: $\chi(a)(\ux) = \chi(a)(\uy)$ if and only if $a \, \ux \, \widetilde{a}^{-1} = a \, \uy \, \widetilde{a}^{-1}$ if and only if $\ux = \uy$. It is surjective since for a given $a \in \Gamma(m)$ and $\ux \in \BR^1_m$ we find that
$$
\chi(a)\left( a^{-1}Ê\ux \, \widetilde{a} \right) = \ux,
$$
where $ a^{-1}Ê\ux \, \widetilde{a} \in \BR^1_m$ because it is equal to $\chi(a^{-1})(\ux)$, $a^{-1} \in \Gamma(m)$. 

\medskip \noindent Finally, the map $\chi(a)$ is an isometry, meaning that $|\ux|^2 = \left|Ê\chi(a)(\ux) \right|^2$:
\begin{align*}
\left| a \, \ux \, \widetilde{a}^{-1} \right|^2 &= \left( a \, \ux \, \widetilde{a}^{-1}Ê\right) \left( a \, \ux \, \widetilde{a}^{-1} \right)^\ast = a \, \ux \, \widetilde{a}^{-1}Ê\, \left( \widetilde{a}^\ast\right)^{-1}Ê\,\ux^\ast \, a^\ast \\
&= a\,\ux \left( \widetilde{a^\ast \, a} \right)^{-1}Ê\, \ux \, a^\ast = |a|^{-2}Ê\, a \, \ux^2 \, a^\ast = |a|^{-2}Ê\, |\ux|^2 \, |a|^2 = |\ux|^2. 
\end{align*}
\end{proof}

\begin{defi}
We define the unit sphere $S^{m-1}$ to be the subspace of $\BR^1_m$ containing unit vectors, i.e. $\uomega \in \BR^1_m$ such that $\uomega^2 = |\uomega|^2 = 1$.
\end{defi}

\begin{lemma}
Take $\uomega \in S^{m-1}$, then $\chi(\uomega)(\ux)[1]$ is the orthogonal reflection with respect to the hyperplane $\uomega^\perp[1]$. 
\end{lemma}
\begin{proof}
From the previous lemma, we know that 
$$
\chi(\uomega)(\ux) = -\frac{1}{|\uomega|^2} \, \uomega \, \ux \, \uomega = -2\, \langle \ux, \uomegaÊ\rangle \, \uomega + \ux.
$$
If we decompose $\uxÊ= \sum_{j=1}^m x_j \, R_j$ as $\ux = \lambda \, \uomega + \un{t}$ with $\lambda \in \BR$ and $\un{t} \in \uomega^\perp$, then $\left\langle \ux, \uomega \right\rangle = \lambda \, |\uomega|^2  = \lambda$ and hence
$$
\chi(\uomega)(\ux) = -2 \,\lambda \, \uomega + \lambda \, \uomega + \un{t} = -\lambda \, \uomega + \un{t}.
$$
Thus $\chi(\uomega)(\ux)[1] $ is exactly the orthogonal reflection of $\ux[1]$ with respect to the hyperplane $\uomega^\perp[1]$. 
\end{proof}

\begin{defi}
The Pin group is the multiplicative group
$$
\textup{Pin}(m) := \left\{ \prod_{i=1}^n \, \ux_i \ : \ n \in \BN, \ \ux_i \in S^{m-1}, \ \forall i = 1, \ldots, nÊ\right\}.
$$

The Spingroup is the multiplicative group
$$
\textup{Spin}(m) := \left\{ \prod_{i=1}^{2n} \, \ux_i \ : \ n \in \BN, \ \ux_i \in S^{m-1}, \ \forall i = 1, \ldots, 2nÊ\right\}.
$$
\end{defi}
Every element $a$ in $\textup{Pin}(m)$ corresponds to an element $-\chi(a)$ that is the composition of $n$ orthogonal reflections and thus to an element of $\textup{O}(m)$. In fact, since $a$ and $-a$ correspond with the same bijective isometrie $\chi(a)$, we call $\textup{Pin}(m)$ a \textbf{double cover} of $\textup{O}(m)$. 

\medskip Every element $a$ in $\textup{Spin}(m)$ corresponds with an element $\chi(a)$ that is the composition of an \emph{even} number of orthogonal reflections and thus, by Hamilton's theorem, an element of $\textup{SO}(m)$. We call $\textup{Spin}(m)$ a \textbf{double cover} of $\textup{SO}(m)$.

\subsection{Orthogonal Spingroup}\label{sec:orthspingroup}
In a completely similar fashion, we can define the linear vectorspace $\BR_m^{1,\perp}$ of vectors in the operators $S_j \, e_j^\perp$, $j= 1,\ldots,m$:
$$
\BR_m^{1,\perp} = \left\{ \uomega = \sum_{j=1}^m \omega_j \, S_j \, e_j^\perp \ : \ \omega_j \in \BR,  \ j = 1,\ldots,m \right\}
$$
and the unit sphere $S^{m-1,\perp} = \left\{ \uomega \in \BR_m^{1,\perp} \ : \ \uomega^2 = -| \uomega|^2 = -1 \right\}$. 

Due to the relation $\left\{ S_j \, e_j^\perp, S_k \, e_k^\perp \right\} = -2\,\delta_{j,k}$, the vectors $\uomega \in \BR_m^{1,\perp}$ satisfy $\uomega^2 = - \sum_{j=1}^m \omega_j^2 = - |\uomega|^2$. For a non-zero vector $\uomega$ of $\BR_m^{1,\perp}$, the inverse is given by $\uomega^{-1} = - \frac{ \uomega}{|\uomega|^2}$. 

Let $\Gamma^\perp(m)$ be the associated Clifford group
$$
\Gamma^\perp(m) = \left\{ \prod_{i=1}^n \uomega_i \ : \ n \in \BN, \ \uomega_i \in \widetilde{\BR}^1_m \backslash \{ 0 \}Ê\right\}.
$$
Again one can associate a bijective isometry $\chi(a) \in O(m)$ with every element $a \in \Gamma^\perp(m)$:
$$
a \in \Gamma^\perp(m) \mapsto \chi(a): \ \BR_m^{1,\perp} \to \BR_m^{1,\perp},\qquad \chi(a)(\ux) = a \, \ux \, \widetilde{a}^{-1}.
$$

In particular, if $a$ is a unit vector $\uomega \in S^{m-1,\perp}$, i.e. $\uomega^2 = -|\uomega|^2 = -1$, then $\chi(\uomega)$ is the orthogonal reflection with respect to the hyperplane $\uomega^\perp$, and as such, the Spingroup
$$
\textup{Spin}^\perp(m) = \left\{ \prod_{i=1}^{2n}Ê\uomega_i \ : \ n \in \BN, \ \uomega_i \in S^{m-1,\perp}, \ \forall\, i = 1, \ldots, 2n \right\}
$$
is a double cover of $\textup{SO}(m)$. 

\medskip In the following section we will now introduce a $\textup{Spin}(m)$-representation within the space of discrete polynomials in the $m$ variables $\xi_1,\ldots,\xi_m$.

\section{$\textup{Spin}(m)$-representation}\label{sec:repr}
Consider the following actions of $s \in \textup{Spin}(m)$ on discrete Clifford-algebra valued polynomials $f\left( \xi_1, \ldots,\xi_m\right)$:
\begin{align*}
H^1(s) \, f( \xi_1,\ldots,\xi_m) &= s \, f(\bar{s} \, \xi \, s) \,\bar{s}, \\
H^0(s) \, f(\xi_1, \ldots, \xi_m) &= f(\bar{s} \, \xi \, s), \\
L(s) \, f(\xi_1, \ldots, \xi_m) &= s \, f(\bar{s} \, \xi \, s).
\end{align*}
We will show that the operators $\p$ and $\Delta$ are $L(s)$-, resp. $H^1(s)$- and $H^0(s)$-invariant. 

\begin{rem}
Both $H$-actions are $\Delta$-invariant and preserve the space $\mathcal{H}_k$, for $k \in \BN$. However, the difference lies in which function space they preserve. As in Euclidean Clifford analysis, the $H^1$-action of the classical Spingroup preserves $k$-vector, we would expect the $H^1$-action of the discrete Spingroup to show a similar treat. This will be a topic for further research. 
\end{rem}

\begin{rem}
Because every discrete function can be expressed by its Taylor series, i.e. in terms of discrete polynomials, this action is readily extendable to all discrete functions. 
\end{rem}

We now explicitly prove that the discrete Dirac operator is invariant under the $L(s)$-action of $s \in \textup{Spin}(m)$. We will start with some auxiliary lemmas
\begin{lemma} 
Let $\uomega = \sum_{j=1}^m \omega_j \, R_j \in \BR^1_m$, then 
$$
\left\langle \p, \uomega \right\rangle = \sum_{j=1}^m \omega_j \, \p_j \, R_j, \qquad 
\left\langle \xi, \uomega \right\rangle = \sum_{j=1}^m \omega_j \, \xi_j \, R_j.
$$
\end{lemma}
\begin{proof}
This follows from the definition and the commutator relations $\left\{Ê\xi_j, R_k \right\} = 2\,\delta_{j,k} \, \xi_j \, R_j$ and $\left\{Ê\p_j, R_k \right\} = 2\,\delta_{j,k} \, \p_j \, R_j$:
\begin{align*}
\left\langle \xi, \uomega\right\rangle &= \frac{1}{2} \, \sum_{j=1}^m \sum_{k=1}^m \omega_k \left\{ \xi_j, R_k\right\} = \sum_{j=1}^m \omega_j \, \xi_j \, R_j, \\
\left\langle \p, \uomega\right\rangle &= \frac{1}{2} \, \sum_{j=1}^m \sum_{k=1}^m \omega_k \left\{ \p_j, R_k\right\} = \sum_{j=1}^m \omega_j \, \p_j \, R_j.
\end{align*} 
\end{proof}
Note that both $\xi_j$ and $R_j$ commute with $\left\langle \xi, \uomega \right\rangle$, for all $j = 1, \ldots, m$, and thus also $\xi$ and $\uomega$ commute with  $\left\langle \xi, \uomega \right\rangle$. 

\begin{cor}
The discrete Dirac operator $\p$ and vector variable $\xi$ admit the following decompositions:
$$
\p = \sum_{j=1}^m \left\langle \p, R_j \right\rangle R_j, \qquad \xi = \sum_{j=1}^m \left\langle \xi, R_j \right\rangle R_j.
$$
\end{cor}
\begin{proof}
This follows from $\p = \sum_{j=1}^m \p_j$, $\left\langle \p, R_jÊ\right\rangle = \p_j \, R_j$ and $R_j^2 = 1$. Similarly for the second statement.
\end{proof}

\begin{lemma}\label{lem:eta}
Let $\uomega \in S^{m-1}$ be a unit vector, i.e. $\uomega^2 = 1$. Define the operator $\eta = \uomega \, \xi \, \uomega = \uomega \left( \xi_1 + \ldots + \xi_m \right) \uomega$. Then $\eta = \sum_{j=1}^m \eta_j$ where 
$$
\eta_j = 2  \left\langle \xi, \uomega \right\rangle \omega_j \, R_j - \xi_j =  \left\langle \eta, R_j \right\rangle R_j
$$
and $\left\{ R_j, \eta_s \right\} = 2\,R_j \, \eta_j \, \delta_{j,s}$. 
\end{lemma}
\begin{proof}
Consider $\eta = \uomega \, \xi \, \uomega$. Since $\xi \, \uomega = 2 \left\langle \xi, \uomega \right\rangle - \uomega \, \xi$ and since $\uomega$ is a unit vector, we get
$$
\eta = \uomega \, \xi \, \uomega = 2 \left\langle \xi, \uomega \right\rangle \uomega - \xi. 
$$
Define $\eta_j = 2  \left\langle \xi, \uomega \right\rangle \omega_j \, R_j - \xi_j$ then obviously $\eta = \sum_{j=1}^m \eta_j$. Furthermore, since $R_j$ commutes with $\left\langle \xi, \uomega \right\rangle$ we find that
\begin{align*}
\left\langle \eta, R_j \right\rangle &= \frac{1}{2}Ê\left( 2 \left\langle \xi, \uomega \right\rangle \uomega \, R_j - \xi \, R_j + 2 \, R_j \left\langle \xi, \uomega \right\rangle \uomega - R_j \, \xi\right) \\
&= 2 \left\langle \xi, \uomega \right\rangle \left\langle \uomega, R_j \right\rangle - \left\langle \xi, R_j \right\rangle 
= 2 \left\langle \xi, \uomega \right\rangle \omega_j - \xi_j \, R_j \\
&= \eta_j \, R_j. 
\end{align*}
Hence
$$
\eta = \sum_{j=1}^m \eta_j  = \sum_{j=1}^m \left\langle \eta, R_j \right\rangle R_j. 
$$
Finally, we note that
\begin{align*}
R_j \, \eta_j &= 2 \, R_j \left\langle \xi, \omega \right\rangle \, \omega_j \, R_j - R_j \, \xi_j = 2 \left\langle \xi, \omega \right\rangle \, \omega_j \,R_j^2 - \xi_j\, R_j = \eta_j \, R_j, \\
R_j \, \eta_s &= 2 \, R_j \left\langle \xi, \omega \right\rangle \, \omega_s \, R_s - R_j \, \xi_s = 2 \left\langle \xi, \omega \right\rangle \, \omega_j \,R_j \, R_s + \xi_s \, R_j = -\eta_j \, R_s.
\end{align*}
\end{proof}

\begin{cor}
Let $\uomega_1$, \ldots, $\uomega_k \in S^{m-1}$ be $k$ unit vectors. Let  $\eta = \uomega_k \, \ldots \uomega_1 \, \xi \, \uomega_1 \ldots \uomega_k$ and $\eta' = \uomega_{k-1}Ê\ldots \uomega_1 \, \xi \, \uomega_1 \ldots \uomega_{k-1}$. Then $\eta = \sum\limits_{j=1}^m \eta_j$ where
$$
\eta_j = 2 \left\langle \eta', \uomega_k \right\rangle \uomega_{k,j} \, R_j - \eta_j' = \left\langle \eta, R_j \right\rangle R_j.
$$
Furthermore $\left\{ R_j, \eta_s \right\} = 2\,R_j \, \eta_j \, \delta_{j,s}$.
\end{cor}
\begin{proof}
We will prove this by induction on $k$. For $k=1$, this is the result of Lemma \ref{lem:eta}. So choose $k > 1$ and suppose the statement holds for $k-1$. Then, by definition,
\begin{align*}
\eta &= \uomega_k \, \eta' \, \uomega_k = 2 \left\langle \eta', \uomega_kÊ\right\rangle \uomega_k - \eta' \, \uomega_k^2 = 2\,\left\langle \eta', \uomega_kÊ\right\rangle \uomega_k - \eta'. 
\end{align*}
Now $2 \left\langle \eta, R_j \right\rangle = \eta \, R_j + R_j \, \eta = \uomega_k \, \eta' \, \uomega_k \,R_j + R_j \, \uomega_k \, \eta' \, \uomega_k$. Since
$$
\uomega_k \, R_j = 2 \left\langle \uomega_k, R_j \right\rangle - R_j \, \uomega_k = 2 \, \omega_{k,j}Ê- R_j \, \uomega_k,
$$
we get that 
\begin{align*}
2 \left\langle \eta, R_j \right\rangle &= \uomega_k \, \eta' \left( 2\,\omega_{k,j} - R_j \, \uomega_k \right) + \left( 2 \, \omega_{k,j}Ê- \uomega_k \, R_j \right) \eta' \, \uomega_k \\
&= 4\,\omega_{k,j} \left\langle \eta', \uomega_kÊ\right\rangle - 2\,\uomega_k \left\langle \eta', R_j \right\rangle \uomega_k.
\end{align*}
As by induction it holds that $R_j \, \eta_j' = \eta_j' \, R_j $ and $R_j \, \eta_s' = - \eta_s' \, R_j$, $j \neq s$, we see that 
$$
R_s \, \left\langle \eta', R_j \right\rangle = R_s \, \eta'_j \, R_j = \eta'_j \, R_j\, R_s = \left\langle \eta', R_j \right\rangle R_s
$$
and thus also $\uomega_k \, \left\langle \eta', R_j \right\rangle = \left\langle \eta', R_j \right\rangle  \uomega_k$. This, combined with $\uomega_k^2 = 1$, implies that
\begin{align*}
\left\langle \eta, R_j \right\rangle &= 2\,\omega_{k,j} \left\langle \eta', \uomega_kÊ\right\rangle - \left\langle \eta', R_j \right\rangle.
\end{align*}
Again from the induction hypothesis, we find that 
\begin{align*}
\left\langle \eta, R_j \right\rangle &= 2 \left\langle \eta', \uomega_kÊ\right\rangle \omega_{k,j} - \eta'_j \, R_j = \eta_j \, R_j.
\end{align*}
We arrive at $\eta_j = \left\langle \eta, R_j \right\rangle R_j$. 

\bigskip Now consider 
\begin{align*}
R_j \, \eta_j &= 2\,R_j \left\langle \eta', \uomega_k \right\rangle \omega_{k,j}Ê\, R_j - R_j\, \eta'_j 
= 2\,R_j \,\sum_{s=1}^m \omega_{k,s} \left\langle \eta', R_s \right\rangle \omega_{k,j}Ê\,R_j - \eta'_j \, R_j \\
&= 2\,\sum_{s=1}^m \omega_{k,s} \left\langle \eta', R_s \right\rangle \omega_{k,j}Ê\,R_j^2 - \eta'_j \, R_j 
= 2  \left\langle \eta', \uomega_k \right\rangle \omega_{k,j}Ê\,R_j^2 - \eta'_j \, R_j \\
&= \eta_j \, R_j. 
\end{align*}
Finally,
\begin{align*}
R_j \, \eta_p &= 2\,R_j \left\langle \eta', \uomega_k \right\rangle \omega_{k,p}Ê\, R_p - R_j\, \eta'_p
= 2\,R_j \,\sum_{s=1}^m \omega_{k,s} \left\langle \eta', R_s \right\rangle \omega_{k,p}Ê\,R_p + \eta'_p \, R_j \\
&= 2\,\sum_{s=1}^m \omega_{k,s} \left\langle \eta', R_s \right\rangle \omega_{k,p}Ê\,R_j \, R_p + \eta'_j \, R_j 
= -2 \left\langle \eta', \uomega_k \right\rangle \omega_{k,p}Ê\,R_p \, R_j + \eta'_j \, R_j \\
&= -\eta_p \, R_j. 
\end{align*}
\end{proof}

The notation $f( \bar{s} \,\xi \, s)$ in the definition of the Spingroup-action means that in the polynomial expression of $f$, we replace each $\xi_j$ by the corresponding $\eta_j$, $j=1,\ldots,m$. 

\begin{lemma}
If we define likewise $\p_\eta = \uomega_k \ldots \uomega_1 \, \p \, \uomega_1 \ldots \uomega_k$ and 
$\p_{\eta'} = \uomega_{k-1} \ldots \uomega_1 \, \p \, \uomega_1 \ldots \uomega_{k-1}$ then, in a similar fashion, we get that 
$$
\p_\eta = 2 \left\langle \p_{\eta'}, \uomega_k \right\rangle \uomega_k - \p_{\eta'} = \sum_{j=1}^m \p_{\eta_j}
$$
where
$$
\p_{\eta_j} = 2 \left\langle \p_{\eta'}, \uomega_k \right\rangle \omega_{k,j} \, R_j - \p_{\eta'_j} = \left\langle \p_\eta, R_j \right\rangle R_j.
$$
\end{lemma}

\begin{lemma}
Let $\eta = \uomegaÊ\, \xi \, \uomega$ and $\eta_j = 2 \left\langle \xi, \uomega\right\rangle \uomega_j \, R_j - \xi_j$. If $f\left( \xi_1, \ldots, \xi_m \right)$ is a discrete monogenic function then also 
$$
\p \left( \uomega\, f\left( \eta_1, \ldots, \eta_m\right)\right) = 0.
$$
\end{lemma}
\begin{proof}
Since $\p_{\eta} = \uomega \, \p \, \uomega$ and $\uomega$ is a unit vector, we see that
$$
\p \, \uomega \, f\left( \eta_1, \ldots, \eta_m \right) = \uomega \, \p_{\eta}Ê\, f\left(\eta_1, \ldots, \eta_m \right) = \uomega \, \sum_{j=1}^m \p_{\eta_j}Ê\, f\left(\eta_1, \ldots, \eta_m \right). 
$$
Now we just have to check that the calculation rules for $\p_{\eta_j}$ with $\eta_k$ are the same as the calculation rules of $\p_j$ and $\xi_j$, i.e. we have to check that 
$$
\left[Ê\p_j, \xi_j \right] = 1, \qquad \left\{ \p_j, \xi_k\right\} = 0, j \neq k. 
$$
We use the explicit formulas for $\p_{\eta_j}$ and $\eta_j$ and the calculation rules for the $\p_j$ and $\xi_j$:
\begin{align*}
\p_{\eta_j}Ê\, \eta_j &= \left( 2 \left\langle \p, \uomega \right\rangle \omega_j \, R_j - \p_j \right) \left( 2  \left\langle \xi, \uomega \right\rangle \omega_j \, R_j - \xi_j\right) \\
&= 4 \left\langle \p, \uomega \right\rangle \omega_j \, R_j \left\langle \xi, \uomega \right\rangle \omega_j \, R_j - 2 \left\langle \p, \uomega \right\rangle \omega_j \, R_j \, \xi_j - 2\, \p_j \left\langle \xi, \uomega \right\rangle \omega_j \, R_j + \p_j \, \xi_j \\
&= 4 \, \omega_j^2 \left\langle \p, \uomega \right\rangle \left\langle \xi, \uomega \right\rangle - 2\,\omega_j \, R_j \left\langle \p, \uomega \right\rangle \xi_j - 2\,\omega_j \, R_j \, \p_j \left\langle \xi, \uomega \right\rangle + \left(1 + \xi_j \, \p_j \right)
\end{align*}
Now 
\begin{align*}
\left\langle \p, \uomega \right\rangle \left\langle \xi, \uomega \right\rangle &= \sum_{j=1}^m \omega_j \, \p_j \, R_j \,\sum_{s=1}^m \omega_s \, R_s \, \xi_s \\
&= \sum_{j=1}^m \omega_j^2 \left( 1 + \xi_j \, \p_j\right) + \sum_{j \neq s}Ê\omega_j \, \omega_s \, R_s \, \xi_s \, R_j \, \p_j \\
&= \left\langle \xi, \uomega \right\rangle \left\langle \p, \uomega \right\rangle+ \sum_{j=1}^m \omega_j^2 
= \left\langle \xi, \uomega \right\rangle \left\langle \p, \uomega \right\rangle + 1, \\
\left\langle \p, \uomega \right\rangle \xi_j &= \sum_{s=1}^m \omega_s \, R_s \, \p_s \, \xi_j = \xi_j \, \sum_{s \neq j}Ê\omega_s \, R_s \, \p_s + \omega_j \, R_j \left( \xi_j \, \p_j + 1\right) \\
&= \xi_j \left\langle \p, \uomega \right\rangle + \omega_j \, R_j, \\
\p_j \left\langle \xi, \uomega \right\rangle &= \sum_{s \neq j}Ê\omega_s \, R_s \, \xi_s \, \p_j + \omega_j \, R_j \left( 1 + \xi_j \, \p_j \right) = \left\langle \xi, \uomega \right\rangle \, \p_j + \omega_j \, R_j.
\end{align*}
Hence
\begin{align*}
\p_{\eta_j}Ê\, \eta_j &= 4 \, \omega_j^2 \left( \left\langle \xi, \uomega \right\rangle \left\langle \p, \uomega \right\rangle + 1\right) - 2\,\omega_j \, R_j \left( \xi_j \left\langle \p, \uomega \right\rangle + \omega_j \, R_j \right) \\
&\phantom{=} - 2\,\omega_j \, R_j \left(  \left\langle \xi, \uomega \right\rangle \, \p_j + \omega_j \, R_j\right) + \left(1 + \xi_j \, \p_j \right) \\
&=  4 \, \omega_j^2 \left\langle \xi, \uomega \right\rangle \left\langle \p, \uomega \right\rangle - 2\,\omega_j \, R_j\, \xi_j \left\langle \p, \uomega \right\rangle - 2\,\omega_j \, R_j \left\langle \xi, \uomega \right\rangle \, \p_j  + \xi_j \, \p_j  + 1 \\
&\phantom{=}  - 2\,\omega_j^2 - 2\,\omega_j^2 +  4 \, \omega_j^2 \\
&= \eta_j \, \p_{\eta_j}Ê+ 1.
\end{align*}
Analogously, for $k \neq j$ we find that $\p_{\eta_j}Ê\, \eta_k = - \eta_k \, \p_{\eta_j}$.
Note that
\begin{align*}
\p_{\eta_j}[1] &= 2 \left\langle \p, \uomega \right\rangle \omega_j \, R_j[1] - \p_j[1]Ê= \omega_j \, \p \, \uomega[e_j] + \omega_j \, \uomega \, \p[e_j] - 0 =  0.
\end{align*}
The statement then follows from the monogenicity of $f$. 
\end{proof}

\begin{lemma}
Let $\eta = \uomega_k \ldots \uomega_1Ê\, \xi \, \uomega_1 \ldots \uomega_k$, $\eta' = \uomega_{k-1}Ê\ldots \uomega_1 \, \xi \, \uomega_1 \ldots \uomega_{k-1}$. Denote $\eta_j = 2 \left\langle \eta', \uomega_k \right\rangle \uomega_{k,j} \, R_j - \eta'_j$. If $f\left( \xi_1, \ldots, \xi_m \right)$ is monogenic then also 
$$
\p \left( \uomega_1 \ldots \uomega_k \, f\left( \eta_1, \ldots, \eta_m\right)\right) = 0.
$$
\end{lemma}

\begin{proof}
Since $\p_{\eta} = \uomega_k \, \ldots \uomega_1 \, \p \, \uomega_1 \ldots \uomega_k$ and $\uomega_i$ are unit vectors, we see that
$$
\p \, \uomega_1 \ldots \uomega_k \, f\left( \eta_1, \ldots, \eta_m \right) = \uomega_1 \ldots \uomega_k \, \p_{\eta}Ê\, f\left(\eta_1, \ldots, \eta_m \right) = \uomega \, \sum_{j=1}^m \p_{\eta_j}Ê\, f\left(\eta_1, \ldots, \eta_m \right). 
$$
Now we just have to check that the calculation rules for $\p_{\eta_j}$ with $\eta_k$ are the same as the calculation rules of $\p_j$ and $\xi_j$, i.e. 
$$
\left[Ê\p_{\eta_j}, \eta_j \right] = 1, \qquad \left\{ \p_{\eta_j}, \eta_k\right\} = 0,\  j \neq k. 
$$
We prove this by induction on $k$. For $k=1$, this is the previous lemma. For $k >1$, we assume that
$$
\left[Ê\p_{\eta'_j}, \eta'_j \right] = 1, \qquad \left\{ \p_{\eta'_j}, \eta'_k\right\} = 0,\  j \neq k. 
$$
Using the explicit formulas for $\p_{\eta_j}$ and $\eta_j$ and the calculation rules for the $\p_{\eta'_j}$ and $\eta'_j$ we find: 
\begin{align*}
\p_{\eta_j}Ê\, \eta_j 
&= 4 \, \omega_{k,j}^2 \left\langle \p_{\eta'}, \uomega_k \right\rangle \left\langle \eta', \uomega_k \right\rangle - 2\,\omega_{k,j} \, R_j \left\langle \p_{\eta'}, \uomega_k \right\rangle \eta'_j - 2\,\omega_{k,j} \, R_j \, \p_{\eta'_j} \left\langle \eta', \uomega_k \right\rangle + \left(1 + \eta'_j \, \p_{\eta'_j} \right)
\end{align*}
Similarly, we find that 
\begin{align*}
\left\langle \p_{\eta'}, \uomega_k \right\rangle \left\langle \eta', \uomega_k \right\rangle 
&= \left\langle \eta', \uomega_k \right\rangle \left\langle \p_{\eta'}, \uomega_k \right\rangle + 1, \\
\left\langle \p_{\eta'}, \uomega_k \right\rangle \eta'_j 
&= \eta'_j \left\langle \p_{\eta'}, \uomega_k \right\rangle + \omega_{k,j} \, R_j, \\
\p_{\eta'_j} \left\langle \eta', \uomega_k \right\rangle 
&= \left\langle \eta', \uomega_k \right\rangle \, \p_{\eta'_j} + \omega_{k,j} \, R_j.
\end{align*}
After some calculations it is hence confirmed that  
$$
\p_{\eta_j}Ê\, \eta_j = \eta_j \, \p_{\eta_j} +1, \qquad 
\p_{\eta_j}Ê\, \eta_s = - \eta_s \, \p_{\eta_j}, \qquad j \neq s. 
$$
The statement then follows from the monogenicity of $f$. 
\end{proof}

\begin{cor}
Let $s \in \textup{Spin}(m)$, i.e. $s = \uomega_1  \ldots \uomega_{2k}$, with $\uomega_i \in S^{m-1}$ unit vectors. Let $f$ be a monogenic function in $\xi_1$ , \ldots, $\xi_m$, then 
$$
\p \left( s \, f(\bar{s} \, \xi \, s)\right) = \p \left( \uomega_1 \, \ldots \, \uomega_{2k} \, f\left( \overline{\uomega_{2k}} \, \ldots \overline{\uomega_1} \, \xi \, \uomega_1 \ldots \uomega_{2k} \right) \right) = 0.
$$
\end{cor}

\begin{thm}
Let $s \in \textup{Spin}(m)$ and let $f$ be a harmonic function in $\xi_1$ , \ldots, $\xi_m$, then $H^1(s) f$ and $H^0(s) f$ are harmonic functions. 
\end{thm}
\begin{proof}
Denote $\eta = \bar{s} \, \xi \, s$. Since 
$$
\p_\eta^2 = \bar{s} \, \pÊ\, s \, \bar{s} \, \pÊ\, s = \bar{s} \, \p^2 \, s = \bar{s} \, s \, \Delta = \Delta
$$
we see that $\p_\eta^2$ is scalar and equal to $\Delta$. The statement then immediately follows from the calculation rules of $\eta_j$ and $\p_{\eta_k}$. 
\end{proof}
\begin{rem}
Note that for a discrete monogenic Clifford-valued function $f$ and for $s \in \textup{Spin}(m)$, the function $H^1(s) f = s \,f(\bar{s}\,\xi\,s) \, \bar{s}$ is not only harmonic but also monogenic. 
\end{rem}

\begin{rem}
Completely similar, one can define the $H^\perp$- and $L^\perp$-actions corresponding to the $\textup{Spin}^\perp(m)$ group: for $s \in \textup{Spin}^\perp(m)$ and $f$ a discrete polynomial, consider the $H^{j,\perp}(s)$ and $L^\perp(s)$-action
\begin{align*}
H^{1,\perp}(s) \, f(\xi) &= s \, f( \bar{s} \, \xi \, s) \, \bar{s}, \\
H^{0,\perp}(s) \, f(\xi) &= f( \bar{s} \, \xi \, s), \\
L^{\perp}(s) \, f(\xi) &= s \, f(\bar{s} \, \xi \, s),
\end{align*}
which are $\Delta$- resp. $\p$-invariant. 
\end{rem}

In the next section we will show the connection of the discrete Spingroup $\textup{Spin}(m)$ and its actions on discrete polynomials and the previously (see \cite{Rotations,Translations}) defined $\mathfrak{so}(m)$-generators $L_{a,b}$ resp. $dR(e_{a,b})$ by determining the infinitesimal representation of the Lie algebra connected to the Spingroup.

\subsection{Infinitesimal representation}\label{sec:infrep}
The following lemma can be proven completely similar to the proof given in \cite[p.16]{Friedrich}.
 
\begin{lemma}
The linear subspace $\mathfrak{b} = \textup{span}_{\BR}Ê\left\{ R_k R_j: 1 \leqslant k < j \leqslant m\right\}$ equipped with the commutator $\left[ x, y \right] = xy -yx$ is a Lie algebra which coincides with the Lie algebra $\mathfrak{spin}(m)$ of the Spingroup $\textup{Spin}(m)$. The exponential map $\exp: \mathfrak{b} \to \textup{Spin}(m)$ is given by 
$$
\exp(x) = \sum_{j=0}^\infty \frac{x^j}{j!}.
$$
\end{lemma}

\begin{lemma}
The associated $\mathfrak{spin}(m)$ Lie algebra-representation $dH^0$ of the Spingroup-representation $H^0$ is given by 
$$
dH^0(R_{ji}) = -2\,L_{ij}Ê=  -2\, R_j \, R_i \left( \xi_i \, \p_j + \xi_j \, \p_i\right), \qquad i < j. 
$$
\end{lemma}
\begin{proof}
To determine the associated algebra-representation, we consider $e^{\varepsilon \, b}$ with $b \in \mathfrak{spin}(m)$ an element of the Lie algebra and $\varepsilon \approx 0$, i.e. we assume $\varepsilon^2 = 0$. Then $e^{\varepsilon \, b} = 1 + \varepsilon \, b$. It suffices to consider the generators $R_{ji}$, $i < j$, of the linear space $\mathfrak{spin}(m)$; let $f(\xi)$ be a discrete polynomial in the variables $\xi_1, \ldots, \xi_m$, then:
\begin{align*}
dH^0(R_{ji}) f(\xi) &= \lim_{\varepsilon \to 0} \frac{1}{\varepsilon} \left( dH^0(R_{ji}) f(\xi) - f(\xi) \right) \\
&= \lim_{\varepsilon \to 0} \frac{1}{\varepsilon} \left( f \left( \overline{\left( 1 + \varepsilon \,R_{ji}\right)} \xi \left( 1 + \varepsilon \,R_{ji}\right) \right) - f(\xi) \right) \\
&= \lim_{\varepsilon \to 0} \frac{1}{\varepsilon} \left( f \left( \left( 1 - \varepsilon \,R_{ji}\right) \xi \left( 1 + \varepsilon \,R_{ji}\right) \right) - f(\xi) \right) \\
&= \lim_{\varepsilon \to 0} \frac{1}{\varepsilon} \left( f \left( \xi - \varepsilon \left[ R_{ji}, \xi \right] \right) - f(\xi) \right).
\end{align*}
Now $\left[ R_{ji}, \xi \right] = R_{ji}Ê\, \xi_i - \xi_i \, R_{ji}Ê+ R_{ji}\, \xi_j - \xi_j \, R_{ji} = 2 \, R_{ji}Ê\left( \xi_i + \xi_j \right)$ and hence 
\begin{align*}
dH^0(R_{ji}) f(\xi) &= \lim_{\varepsilon \to 0} \frac{1}{\varepsilon} \left( f \left( \xi - 2\,\varepsilon \, R_{ji}Ê\left( \xi_i + \xi_j \right) \right) - f(\xi) \right).
\end{align*}
The evaluation of $f$ in $\xi - 2\,\varepsilon \, R_{ji}Ê\left( \xi_i + \xi_j \right)$ implies that every $\xi_s$ is replaced by the part of $\xi - 2\,\varepsilon \, R_{ji}Ê\left( \xi_i + \xi_j \right)$ that commutes with $\xi_s$, i.e.
$$
\xi_i \mapsto \xi_i - 2\,\varepsilon \, R_{ji}Ê\, \xi_j, \qquad \xi_j \mapsto \xi_j - 2\,\varepsilon \, R_{ji} \, \xi_i, \qquad \xi_s \mapsto \xi_s \quad (s \neq i, j).
$$
If we decompose the polynomial $f$ as $\sum_{\alpha}Ê\xi_1^{\alpha_1} \ldots \xi_m^{\alpha_m} \, \lambda_\alpha$ with $\lambda_\alpha$ Clifford-valued constants, then it follows from $\varepsilon^2 = 0$ that
\begin{align*}
& dH^0(R_{ji}) f(\xi) \\
&= \sum_{\alpha}Ê\,  \lim_{\varepsilon \to 0} \frac{1}{\varepsilon} \left( \xi_1^{\alpha_1} \ldots \left( \xi_i - 2\,\varepsilon \, R_{ji}Ê\, \xi_j\right)^{\alpha_i}Ê\ldots \left( \xi_j - 2\,\varepsilon \, R_{ji} \, \xi_i\right)^{\alpha_j} \ldots \xi_m^{\alpha_m} - \xi_1^{\alpha_1}Ê\ldots \xi_i^{\alpha_i}Ê\ldots \xi_j^{\alpha_j} \ldots \xi_m^{\alpha_m} \right) \lambda_\alpha \\
&= \sum_{\alpha}Ê\,  \lim_{\varepsilon \to 0} \frac{1}{\varepsilon} \left( \xi_1^{\alpha_1} \ldots \left( \xi_i^{\alpha_i}Ê- 2\,\varepsilon \,\alpha_i\, \xi_i^{\alpha_i-1} \, R_{ji}Ê\, \xi_j\right)Ê\ldots \left( \xi_j^{\alpha_j} - 2\,\varepsilon \,\alpha_j\, \xi_j^{\alpha_j-1}Ê\, R_{ji} \, \xi_i\right) \ldots \xi_m^{\alpha_m} \right. \\
&\phantom{=======}Ê\left. - \xi_1^{\alpha_1}Ê\ldots \xi_i^{\alpha_i}Ê\ldots \xi_j^{\alpha_j} \ldots \xi_m^{\alpha_m} \right) \lambda_\alpha \\
&= - 2\, \sum_{\alpha} \left( \alpha_j \, \xi_1^{\alpha_1} \ldots \xi_i^{\alpha_i}Ê\ldots \left(\xi_j^{\alpha_j-1}Ê\, R_{ji} \, \xi_i\right) \ldots \xi_m^{\alpha_m} + \alpha_i \, \xi_1^{\alpha_1} \ldots \left( \xi_i^{\alpha_i-1} \, R_{ji}Ê\, \xi_j\right)Ê\ldots \xi_j^{\alpha_j} \ldots \xi_m^{\alpha_m} \right) \lambda_\alpha \\
%
%
%
&= - 2\, \sum_{\alpha} \left( R_j \, R_i \left( \xi_i \,\p_j + \xi_j\,\p_i \right) \xi_1^{\alpha_1} \ldots \xi_i^{\alpha_i}Ê\ldots \xi_j^{\alpha_j} \ldots \xi_m^{\alpha_m} \right) \lambda_\alpha \\
&= - 2\, L_{ij}\, \sum_{\alpha} \left( \xi_1^{\alpha_1} \ldots \xi_i^{\alpha_i}Ê\ldots \xi_j^{\alpha_j}  \ldots \xi_m^{\alpha_m} \right) \lambda_\alpha.
\end{align*}
Since any discrete function can be decomposed into polynomials by means of its discrete Taylor decomposition, this proves the statement. 
\end{proof}

\begin{lemma}
The associated $\mathfrak{spin}(m)$ algebra-representation $dH^1$ resp. $dL$ of the $\textup{Spin}(m)$-representation $H^1$ resp. $L$ is given by 
\begin{align*}
dH^1(R_{ji}) &= -2\left( L_{ij} - \frac{1}{2}Ê\,R_{ji}Ê\right), & 
dL(R_{ji}) &= -2\, L_{ij}Ê+ \left[ÊR_{ji}, \cdot \right]. 
\end{align*}
\end{lemma}

\begin{lemma}
The maximal torus, i.e. the maximal abelian subgroup, of $\textup{Spin}(m)$ is given by
$$
\mathbb{T}= \left\{ \exp\left( \frac{1}{2}Ê\left( t_1 \, R_2\,R_1 + \ldots + t_n \, R_{2n}Ê\, R_{2n-1}Ê\right) \right): \ t_i \in \BR \right\}.
$$ 
Here $n = \left\lfloor \frac{m}{2}\right\rfloor$. 
\end{lemma}
\begin{proof}
The proof is similar to the determination of the maximal torus of the Spingroup in Euclidean Clifford analysis, see for example \cite{rood, Friedrich}. 
\end{proof}

\section{Two-dimensional setting}\label{sec:twodims}
The two-dimensional setting offers is the first non-trivial example with a low value of $m$; we will consider some explicit examples to demonstrate the method. Furthermore, in \cite{Rotations,Hkdecomp} we established the polynomial space $\mathcal{H}_k$ as representation of the Lie algebra $\mathfrak{so}(m,\mathbb{C})$; we will now demonstrate how it works as $\textup{SO}(m)$-representation, namely through the $H^0$-action.

\medskip Consider in two dimensions the space $\BR^1_2 = \left\{ \uomega = \omega_1 \, R_1 + \omega_2 \, R_2: \ \omega_1, \omega_2 \in \BR \right\}$. The Spingroup is then given by 
$$
\textup{Spin}(2) = \left\{ \cos\left(\frac{\theta}{2}\right) + \sin\left(\frac{\theta}{2}\right) R_1 \,R_2: \ \theta \in \BR \right\}. 
$$

Taking $s \in \textup{Spin}(2)$, we see that 
\begin{align*}
\bar{s} \, \xi \, s &= \left( \cos\left(\frac{\theta}{2}\right) - \sin\left(\frac{\theta}{2}\right) \, R_1 \, R_2 \right) \left( \xi_1 + \xi_2 \right) \left( \cos\left(\frac{\theta}{2}\right) + \sin\left(\frac{\theta}{2}\right) \, R_1 \, R_2 \right) \\
%
&= \left( \cos\left(\frac{\theta}{2}\right)^2 - \sin\left(\frac{\theta}{2}\right)^2 \right) \xi_1
+ 2\, \sin \left(\frac{\theta}{2}\right) \cos\left(\frac{\theta}{2}\right) \, \xi_1 \, R_1 \, R_2 \\
&\phantom{=}Ê+ \left( \cos\left(\frac{\theta}{2}\right)^2 - \sin\left(\frac{\theta}{2}\right)^2 \right) \xi_2 + 2\,\sin \left(\frac{\theta}{2}\right) \cos\left(\frac{\theta}{2}\right) \, \xi_2 \, R_1 \, R_2 \\
&= \cos\left(\theta \right) \, \xi_1 + 2\,\sin \left(\theta\right) \, \xi_1 \, R_1 \, R_2 + \cos\left(\theta \right) \, \xi_2 + 2\,\sin\left( \theta\right) \, \xi_2 \, R_1 \, R_2
\end{align*}
Thus, $\eta = \bar{s} \, \xi \, s$ and
%
\begin{align*}
\eta_1 &= \cos\left(\theta \right) \, \xi_1 - \sin\left( \theta\right) \, \xi_2 \, R_2\, R_1, \\
\eta_2 &= \cos\left(\theta \right) \, \xi_2 - \sin\left( \theta\right) \, \xi_1 \, R_2\, R_1.
\end{align*}

\subsection{Examples $H^0$-action}
We will demonstrate the connection between the $H^0$-action of $\textup{Spin}(m)$ and the $\mathfrak{spin}(m)$-action of its associated Lie algebra through some examples. In two dimensions, the $\mathfrak{so}(m)$ Lie algebra-action on a discrete function $f$ is given by, see e.g. \cite{Rotations,Translations}:
$$
f \mapsto \exp\left( \theta \, L_{12}Ê\right) f.
$$
Here $L_{12} = R_2 \, R_1 \left( \xi_2 \, \p_1 + \xi_1 \, \p_2 \right)$. Since $L_{12}^2 = - \left( \xi_1 \, \p_2 + \xi_2 \, \p_1 \right)^2$, this can be rewritten as
$$
f \mapsto \cos\left( \theta \left( \xi_2 \, \p_1 + \xi_1 \, \p_2 \right)Ê\right) f + \sin\left( \theta \left( \xi_2 \, \p_1 + \xi_1 \, \p_2 \right)Ê\right) R_2 \,R_1Ê\,f.
$$
On the other hand, for $f(\xi_1,\xi_2)[1] = \sum_{k, \ell = 0}^\infty \xi_1^k \, \xi_2^\ell[1]Ê\, c_{k,\ell} $ with $c_{k,\ell}$ Clifford-algebra constants, the $H^0$-action of the $\textup{Spin}(m)$-group is given by
$$
f(\xi_1,\xi_2)[1] \mapsto  f\left(\bar{s} \, \xi \, s \right)[1] = \sum_{k,\ell=0}^\infty \, \eta_1^k \, \eta_2^\ell[1] \, c_{k,\ell}.
$$

\begin{example}
Let $f = \xi_1[1]\, c$, $c$ a Clifford-algebra constant, then the Lie algebra-action is given by
\begin{align*}
\cos\left( \theta \left( \xi_2 \, \p_1 + \xi_1 \, \p_2 \right)Ê\right) \xi_1[1] \,c + \sin\left( \theta \left( \xi_2 \, \p_1 + \xi_1 \, \p_2 \right)Ê\right) R_2 \,R_1Ê\, \xi_1[1]\,c.
\end{align*}
Since $ \left( \xi_2 \, \p_1 + \xi_1 \, \p_2 \right) \xi_1[1] = \xi_2[1]$ and $ \left( \xi_2 \, \p_1 + \xi_1 \, \p_2 \right)^2 \xi_1[1] = \xi_1[1]$, we thus find
\begin{align*}
& \cos\left( \thetaÊ\right) \xi_1[1] \, c - \sin\left( \thetaÊ\right) \xi_2\,R_2 \,R_1[1] \, c 
= \cos\left( \thetaÊ\right) \xi_1[1]\,c - \sin\left( \thetaÊ\right) \xi_2[1] \,e_2 \, e_1 \, c.
\end{align*}
For the $H^0$-action, we replace $\xi_1$ by $\eta_1$ and thus find 
\begin{align*}
H^0(s)\left( \xi_1[1]\,c \right) &= \eta_1[1] \,c = \cos\left(\theta \right) \, \xi_1[1] \, c - \sin\left( \theta\right) \, \xi_2 \, R_2\, R_1[1]Ê\, c \\
&= \cos\left( \thetaÊ\right) \xi_1[1]\,c - \sin\left( \thetaÊ\right) \xi_2[1] \,e_2 \, e_1 \, c.
\end{align*}
This is a first illustration of the correspondence of both actions.
\end{example}

We will now reconsider the examples given in \cite{Rotations} and compare those results with the polynomial resulting from the $H^0(s)$-action. We found the following eigenfunctions within $\mathcal{H}_k$ for the action of $\exp\left( \theta \, L_{12}Ê\right)$.

\begin{lemma}
Let $k$ be even, then for $i = 0,\ldots, \frac{k}{2}$, the $k$-homogeneous harmonic functions
\begin{align*}
f_{\pm 2i}^{(k)} &= \xi^{k-2i} \left( \left( \xi_2 \pm \xi_1 \right) \left( \xi_2 \mp \xi_1 \right) \right)^i[1]
\end{align*}
are eigenfunctions of $\xi_1 \,\p_2 + \xi_2 \, \p_1$ with corresponding eigenvalues $\pm 2i$. 

\medskip For $k$ odd and $i=1,\ldots, \frac{k+1}{2}$, the $k$-homogeneous harmonic functions
\begin{align*}
f_{\pm \,(2i-1)}^{(k)} &= \xi^{k-2i+1} \left( \left( \xi_2 \pm \xi_1 \right) \left( \xi_2 \mp \xi_1 \right) \right)^{i-1} \left( \xi_2 \pm \xi_1 \right)[1]
\end{align*}
are eigenfunctions of $\xi_1 \,\p_2 + \xi_2 \, \p_1$ with corresponding eigenvalues $\pm (2i-1)$. 
\end{lemma}
It thus suffices to consider these basis eigenfunctions for future reference. 
 
\begin{example}
There are three eigenfunctions of $\xi_1 \,\p_2 + \xi_2 \, \p_1$ which are homogeneous of degree $2$, namely 
$$
\left(\xi_2^2 + \xi_1^2\right)[1], \qquad \left(\xi_2^2 - 2 \, \xi_1 \, \xi_2 - \xi_1^2\right)[1], \qquad \left(\xi_2^2 + 2 \, \xi_1 \, \xi_2 - \xi_1^2\right)[1].
$$
We will consider these for our next examples. 

\smallskip Let $f = \left(\xi_2^2 - 2 \, \xi_1 \, \xi_2 - \xi_1^2\right)[1]$, then since $\left( \xi_1 \,\p_2 + \xi_2 \, \p_1Ê\right) f = -2\,f$, the Lie algebra-action results in
\begin{align*}
\exp\left( \theta \, L_{12}Ê\right) \xi_1[1] &= \cos\left( 2\,\thetaÊ\right) f - \sin\left( 2\,\thetaÊ\right) R_2 \,R_1Ê\,f.
\end{align*}
The group-action, with $s = \cos(\theta/2) + \sin(\theta/2) \, R_1 \, R_2 \in \textup{Spin}(m)$, is given by
\begin{align*}
& H^0(s)(f) = \left(\eta_2^2 - 2\,\eta_1 \, \eta_2 - \eta_1^2\right)[1] \\
&= \left( \cos(\theta) \, \xi_2 - \sin\left( \theta\right) \, \xi_1 \, R_2\, R_1 \right)^2[1] \\
&\phantom{=} - 2 \left( \cos(\theta) \, \xi_1 - \sin( \theta) \, \xi_2 \, R_2\, R_1 \right) \left( \cos(\theta) \, \xi_2 - \sin( \theta) \, \xi_1 \, R_2\, R_1 \right)[1] \\
&\phantom{=} -  \left( \cos(\theta) \, \xi_1 - \sin( \theta) \, \xi_2 \, R_2\, R_1 \right)^2[1] \\
%
%
&= \left( \cos(\theta)^2 \, \xi_2^2 + \sin( 2\,\theta) \, \xi_1 \, \xi_2 \, R_2\, R_1 + \sin( \theta)^2 \, \xi_1^2 \right)[1] \\
&\phantom{=} - 2 \left( \cos(\theta)^2 \, \xi_1 \, \xi_2 + \sin( \theta) \, \cos(\theta) \, \xi_2^2 \, R_2\, R_1 -  \sin( \theta) \, \cos(\theta) \, \xi_1^2 \, R_2\, R_1 - \sin( \theta)^2 \, \xi_1 \, \xi_2 \right)[1] \\
&\phantom{=} - \left( \cos(\theta)^2 \, \xi_1^2 - \sin(2\, \theta)\, \xi_1 \, \xi_2 \, R_2\, R_1 + \sin( \theta)^2 \, \xi_2^2 \right) [1] \\
&= \cos(2\,\theta) \left( \xi_2^2 - 2 \, \xi_1 \, \xi_2 - \xi_1^2 \right)[1] - \sin( 2\,\theta) \, R_2\, R_1 \left( \xi_2^2 - 2 \, \xi_1 \, \xi_2 - \xi_1^2\right)[1] \\
&= \cos\left( 2\,\thetaÊ\right) f - \sin\left( 2\,\thetaÊ\right) R_2 \,R_1Ê\,f.
\end{align*}

Now let $f = \left(\xi_2^2 + 2 \, \xi_1 \, \xi_2 - \xi_1^2\right)[1]$, then $\left( \xi_1 \,\p_2 + \xi_2 \, \p_1\right) f = 2 \,f$ and hence the Lie algebra-action results in
\begin{align*}
\exp\left( \theta \, L_{12}Ê\right) f &= \cos\left( 2\,\thetaÊ\right) f + \sin\left( 2\,\thetaÊ\right) R_2 \,R_1Ê\,f.
\end{align*}
Adapting the relevant signs in the previous calculation immediately shows that the group-action results in
\begin{align*}
H^0(s)(f) &= \left(\eta_2^2 + 2\,\eta_1 \, \eta_2 - \eta_1^2\right)[1] \\
&= \cos(2\,\theta) \left( \xi_2^2 + 2 \, \xi_1 \, \xi_2 - \xi_1^2 \right)[1] + \sin( 2\,\theta) \, R_2\, R_1 \left( \xi_2^2 + 2 \, \xi_1 \, \xi_2 - \xi_1^2\right)[1] \\
&= \cos\left( 2\,\thetaÊ\right) f + \sin\left( 2\,\thetaÊ\right) R_2 \,R_1Ê\,f.
\end{align*}

Now for $f = \xi_1^2[1] + \xi_2^2[1]$, we found in \cite{Rotations} that $\exp\left( \theta\, L_{12}Ê\right) f = f$. Again considering only the appropriate terms in the previous calculations, we immediately see that also $H^0(s) f = f$.
\end{example}

\subsection{Examples $L$-action}
We will also demonstrate the connection between the L-action of $\textup{Spin}(m)$ and the Lie algebra-action of its associated Lie algebra through some examples. In two dimensions, the $\mathfrak{so}(m)$-action corresponding to the $L$-action is given by 
$$
f \mapsto \exp\left( \theta \, dR(e_{12})Ê\right) f.
$$
Here $dR(e_{12})Ê= L_{12} - \frac{1}{2}Ê\,R_2 \, R_1$. Since $dR(e_{12})^2 = -\left( \xi_1 \, \p_2 + \xi_2 \, \p_1 - \frac{1}{2}Ê\right)^2$, it holds that
$$
f \mapsto \cos\left( \theta \left( \xi_1 \, \p_2 + \xi_2 \, \p_1 - \frac{1}{2}Ê\right)\right) f + \sin\left( \theta \left( \xi_1 \, \p_2 + \xi_2 \, \p_1 - \frac{1}{2}Ê\right)\right) R_2 \,R_1Ê\,f.
$$
Let $f[1] = \sum_{k, \ell = 0}^\infty \xi_1^k \, \xi_2^\ell[1]Ê\, c_{k,\ell} $ with $c_{k,\ell}$ Clifford-algebra constants. 
The L-action of the $\textup{Spin}(m)$-group is given by
$$
f[1] \mapsto  s\, f\left(\bar{s} \, \xi \, s \right)[1] = s \left( \sum_{k,\ell=0}^\infty \, \eta_1^k \, \eta_2^\ell\right)[1] \, c_{k,\ell}.
$$
We will again consider an example from \cite{Rotations}, which consists of an eigenfunction of $\xi_1 \, \p_2 + \xi_2 \, \p_1$, and compare those results with the results from the $\textup{SO}(m)$-action. 

\begin{example}
Let $f = \left(\xi_2 - \xi_1\right)[1]\, c$, $c$ a Clifford-algebra constant. As $\left(\xi_1 \, \p_2 + \xi_2 \, \p_1\right) f = - f$, the Lie algebra-action results in
\begin{align*}
\exp\left( \theta \, dR(e_{12})Ê\right) \xi_1[1]\,c 
&= \cos\left( \frac{3}{2}\,\theta \right) \left(\xi_2 - \xi_1\right)[1]\, c + \sin\left( \frac{3}{2}\,\theta \right) \left(\xi_2 - \xi_1\right)[1] \, e_2 \, e_1 \, c.
\end{align*}

The $L$-action on the other hand is given by
\begin{align*}
& L(s) f = s \left(\eta_2 - \eta_1 \right)[1] \,c \\
&= \left( \cos\left( \frac{\theta}{2}Ê\right) + \sin\left( \frac{\theta}{2}Ê\right) \, R_1 \, R_2 \right) \left( \cos\left(\theta \right) \, \xi_2 - \sin\left( \theta\right) \, \xi_1 \, R_2\, R_1 - \cos\left(\theta \right) \, \xi_1 + \sin\left( \theta\right) \, \xi_2 \, R_2\, R_1\right) [1]Ê\, c \\
&= \cos\left( \frac{\theta}{2}Ê\right) \left( \cos\left(\theta \right) \, \xi_2 - \sin\left( \theta\right) \, \xi_1 \, R_2\, R_1 - \cos\left(\theta \right) \, \xi_1 + \sin\left( \theta\right) \, \xi_2 \, R_2\, R_1\right) [1]Ê\, c \\
&\phantom{=} + \sin\left( \frac{\theta}{2}Ê\right) \left( - \cos\left(\theta \right) \, \xi_2 \, R_1 \, R_2 + \sin\left( \theta\right) \, \xi_1 + \cos\left(\theta \right) \, \xi_1 \, R_1 \, R_2 - \sin\left( \theta\right) \, \xi_2 \right) [1]Ê\, c \\
&= \left( \cos\left( \frac{\theta}{2}Ê\right) \cos\left(\theta \right) - \sin\left( \frac{\theta}{2}Ê\right) \sin\left( \theta\right) \right) \left(\xi_2 - \xi_1\right)[1] \, c \\
&\phantom{=} + \left( \sin\left( \theta\right) \cos\left( \frac{\theta}{2}Ê\right) + \sin\left( \frac{\theta}{2}Ê\right) \cos\left(\theta \right) \right) \left(\xi_2 - \xi_1\right) R_2\, R_1[1]Ê\, c \\
&= \cos\left( \frac{3\,\theta}{2}Ê\right) \left(\xi_2- \xi_1\right)[1] \, c + \sin\left( \frac{3\,\theta}{2}Ê\right) \left( \xi_2- \xi_1 \right)[1]\,e_2 \,e_1Ê\, c.
\end{align*}
Indeed both actions coincide.
\end{example}

\section{Discrete distributions}\label{sec:distr}
We can naturally extend the action of the $\textup{Spin}(m)$-group to discrete distributions: any discrete distribution has a unique dual Taylor series expansion in terms of the discrete derivatives of the delta distribution $\bdelta_0$:
$$
F = \sum_{k_1, \ldots, k_m = 0}^\infty \p_1^{k_1}Ê\ldots \p_m^{k_m} \, \bdelta_0 \  c_{k_1, \ldots, k_m}
$$
with $c_{k_1, \ldots, k_m}$ Clifford algebra constants. Then we consider the $H^0$-action of a Spingroup element $s \in \textup{Spin}(m)$:
$$
H^0(s): \ F \mapsto F(\bar{s} \, \p \, s) = \sum_{k_1, \ldots, k_m = 0}^\infty \p_{\eta_1}^{k_1}Ê\ldots \p_{\eta_m}^{k_m} \, \bdelta_0 \  c_{k_1, \ldots, k_m}.
$$

\subsection{Examples in two dimensions}
In two dimensions, we found in \cite{Rotations} a set of eigendistributions of the algebra-action
$$
\exp\left( \theta \, \bL_{12}\right) = \exp \left( \theta \, R_2 \, R_1 \left( \bxi_1 \, \p_2 + \bxi_2 \, \p_1\right)Ê\right). 
$$
We will give these eigendistributions again explicitly: 
\begin{lemma}
Let $k$ be even, then for $i = 0,1 \ldots, \frac{k}{2}$, the distributions
\begin{align*}
\cB_{\pm \,2i}^{(k)} &= \p^{k-2i} \left( \left( \p_2 \pm \p_1 \right) \left( \p_2 \mp \p_1 \right) \right)^i \bdelta_0
\end{align*}
are eigenfunctions of $\bxi_1 \,\p_2 + \bxi_2 \, \p_1$ with corresponding eigenvalues $\pm 2i$. 

\medskip For $k$ odd and $i=1, \ldots, \frac{k+1}{2}$, the distributions: 
\begin{align*}
\cB_{\pm (2i-1)}^{(k)} &= \p^{k-2i+1} \left( \left( \p_2 \pm \p_1 \right) \left( \p_2 \mp \p_1 \right) \right)^{i-1} \left( \p_2 \pm \p_1 \right) \bdelta_0
\end{align*}
are eigenfunctions of $\bxi_1 \,\p_2 + \bxi_2 \, \p_1$ with corresponding eigenvalues $\pm \,(2i-1)$. 

\medskip The action of a rotation $\exp( \theta\, \bL_{12})$ on these eigenvectors is given by $(i= 0,\ldots,k)$:
\begin{align*}
\exp( \theta\, \bL_{12} )\; \cB^{(2k)}_{\pm \, 2i} &= \left(\cos(2i\,\theta) \pm \sin(2i\,\theta)\; R_2 \, R_1 \right) \, \cB^{(2k)}_{\pm \, 2i},  
\end{align*}
and for $i=1,\ldots,k+1$:
\begin{align*}
\exp( \theta\, \bL_{12} )\; \cB^{(2k+1)}_{\pm \, (2i-1)} &= \left( \cos((2i-1)\,\theta) \pm \sin((2i-1)\,\theta)\; R_2 \, R_1 \right) \, \cB^{(2k+1)}_{\pm \, (2i-1)}.
\end{align*}
\end{lemma}

As $\cB^{(2k)}_{\pm 2i}$ consists of an even number of $\p_1$ and $\p_2$'s, the $R_2 \, R_1$ commutes with $\cB^{(2k)}_{\pm 2i}$; similarly $\cB^{(2k+1)}_{\pm (2i-i)}$ has an odd number of $\p_1$ and $\p_2$'s and will thus anti-commute with $R_2 \, R_1$. Furthermore, since $R_i \, \bdelta_0 = \bdelta_0 \, e_i$, we get that 

\begin{align*}
\exp( \theta\, \bL_{12} )\; \cB^{(2k)}_0 &= \cB^{(2k)}_0, \\[1ex]
\exp( \theta\, \bL_{12} )\; \cB^{(2k)}_{\pm \, 2i} &= \cB^{(2k)}_{\pm \, 2i} \left( \cos(2i\,\theta) \pm \sin(2i\,\theta)\right) e_2 \, e_1, \\[1ex]
\exp( \theta\, \bL_{12} )\; \cB^{(2k+1)}_{\pm \,(2i-1)} &= \cB^{(2k+1)}_{\pm\,(2i-1)} \left( \cos((2i-1)\,\theta) \mp \sin((2i-1)\,\theta)\, e_2 \, e_1\right).
\end{align*}

We can thus compare both actions:
\begin{example}
Consider in two dimensions the discrete distribution $\bdelta_{(1,0)}$:
$$
\bdelta_{(1,0)} = \bdelta_0 - \p_1 \, \bdelta_0 \, e_1 + \frac{1}{2}Ê\, \p_1^2 \, \bdelta_0 \left( 1 + e_1^\perp \, e_1 \right).
$$
If we take a general element $s= \cos\left( \frac{\theta}{2}Ê\right) + \sin\left( \frac{\theta}{2}Ê\right) R_2 \, R_1$ of the Spingroup in two dimensions and define $\p_\eta = \bar{s} \, \p \, s$, then similarly as before, we find that 
$$
\p_{\eta_1} = \cos(\theta) \, \p_1 + \sin(\theta) \, \p_2 \, R_1 \, R_2.
$$
If we thus replace each $\p_1$ by $\p_{\eta_1}$ in the dual Taylor series of $\bdelta_{(1,0)}$, we get the rotated distribution:
\begin{align*}
\textup{Rot}_\theta \, \bdelta_{(1,0)} &= \bdelta_0 - \p_{\eta_1} \, \bdelta_0 \, e_1 + \frac{1}{2}Ê\, \p_{\eta_1}^2 \, \bdelta_0 \left( 1 + e_1^\perp \, e_1 \right) \\
&= \bdelta_0 - \cos(\theta) \, \p_1 \, \bdelta_0 \, e_1 - \sin(\theta) \, \p_2 \, R_1 \, R_2 \, \bdelta_0 \, e_1 \\
&\phantom{=}Ê+ \frac{1}{2}Ê\left( \cos(\theta)^2 \, \p_1^2 + \sin(\theta)^2 \, \p_2^2 + \sin(2\theta) \, \p_1 \, \p_2 \, R_1 \, R_2 \right) \bdelta_0 \left( 1 + e_1^\perp \, e_1 \right) \\
&= \bdelta_0 - \cos(\theta) \, \p_1 \, \bdelta_0 \, e_1 + \sin(\theta) \, \p_2 \, \bdelta_0 \, e_2 \\
&\phantom{=}Ê+ \frac{1}{2}Ê\left( \cos(\theta)^2 \, \p_1^2 + \sin(\theta)^2 \, \p_2^2\right) \bdelta_0 \left( 1 + e_1^\perp \, e_1 \right) + \frac{1}{2}Ê\, \sin(2\theta) \, \p_1 \, \p_2 \, \bdelta_0 \, e_1 \, e_2 \left( 1 + e_1^\perp \, e_1 \right)
\end{align*}

On the other hand, by decomposing $\bdelta_{(1,0)}$ into eigenfunctions of $\bL_{12}$, we can easily write down the algebra-action:
\begin{align*}
\bdelta_{(1,0)} &= \bdelta_0 - \frac{1}{2}Ê\, \cB^{(1)}_1 \, e_1 + \frac{1}{2}Ê\, \cB^{(1)}_2 \, e_1 + \frac{1}{4}Ê\, \cB^{(2)}_1 \left( 1 + e_1^\perp \, e_1 \right) - \frac{1}{8}Ê\left( \cB^{(2)}_2 + \cB^{(2)}_3 \right) \left( 1 + e_1^\perp \, e_1 \right).
\end{align*}
and hence
\begin{align*}
\textup{Rot}_\theta \, \bdelta_{(1,0)} &= \bdelta_0 - \frac{1}{2}Ê\, \cos(\theta)\, \cB^{(1)}_1 \, e_1 + \frac{1}{2}Ê\, \sin(\theta) \, \cB^{(1)}_1 \, e_2 \, e_1 \, e_1 \\
&\phantom{===} + \frac{1}{2}Ê\, \cos(\theta) \, \cB^{(1)}_2 \, e_1 + \frac{1}{2}Ê\, \sin(\theta) \, \cB^{(1)}_2\, e_2 \, e_1 \, e_1 \\
&\phantom{===} + \frac{1}{4}Ê\, \cB^{(2)}_1 \left( 1 + e_1^\perp \, e_1 \right) 
- \frac{1}{8}\, \cos(2\theta)Ê\left( \cB^{(2)}_2 + \cB^{(2)}_3 \right) \left( 1 + e_1^\perp \, e_1 \right) \\
&\hspace{4.5cm} - \frac{1}{8}\,\sin(2\theta)Ê\left( \cB^{(2)}_2 - \cB^{(2)}_3 \right) e_2 \, e_1 \left( 1 + e_1^\perp \, e_1 \right) \\
%
%
%
&= \bdelta_0 - \cos(\theta) \, \p_1 \, \bdelta_0 \, e_1 + \sin(\theta) \, \p_2 \, \bdelta_0 \, e_2 \\
&\phantom{=}Ê+ \frac{1}{2}Ê\left( \cos(\theta)^2 \, \p_1^2 + \sin(\theta)^2 \, \p_2^2\right) \bdelta_0 \left( 1 + e_1^\perp \, e_1 \right) + \frac{1}{2}Ê\, \sin(2\theta) \, \p_1 \, \p_2 \, \bdelta_0 \, e_1 \, e_2 \left( 1 + e_1^\perp \, e_1 \right).
\end{align*}
This again demonstrates the correspondence between the group action $H^0(s)$ and the associated $\mathfrak{spin}(m)$-action given by $\exp(\theta \, \bL_{12})$. 
\end{example}

\section{Irreducible Representations of $\textup{Spin}(m)$}\label{sec:irrep}

\subsection{The fundamental representations with integer valued highest weights}
In this section, we take a look at the irreducible representations corresponding to the fundamental weight $\left(k\right)$. Consider thus the space $\mathcal{H}_k$ of Cliffordalgebra-valued homogeneous harmonic polynomials of degree $k$. Every element $H_k \in \cH_k$ can be decomposed according to the basis elements $\e_j^\pm$, $j =1,Ê\ldots,m$, or equivalently, according to the operators $R_j$ and $S_j \, e_j^\perp$. In other words, $\cH_k = \textup{Alg}_\BC\left\{ \xi_i, R_i, S_i \, e_i^\perp, \ i = 1,\ldots,m\right\}$. We will now restrict ourselves to the subalgebra $[\cH_k]_0$ of $\cH_k$, consisting of those elements for which each $\xi_i$ is accompanied by the corresponding $R_i$, and where no $S_i\,e_i^\perp$ appear:
$$
[\cH_k]_0 = \textup{Alg}_\BC\left\{ \xi_i \, R_i, \ i = 1,\ldots,m\right\}.
$$
We will show that this is an irreducible $\textup{Spin}(m)$-representation with highest weight $(k)$.

\begin{lemma}
Let $[\cP_k]_0$ be the space of discrete homogeneous polynomials of degree $k$, with complex coefficients, in $m$ variables $\xi_1\,R_1, \ldots, \xi_m\,R_m$. Then 
$$
[\cP_k]_0 = [\cH_k]_0 \oplus \xi^2 \, [\cP_{k-2}]_0. 
$$
\end{lemma}
\begin{proof}
Note that $[\cH_k]_0 \subset [\cP_k]_0$ and $\xi^2 \, [\cP_{k-2}]_0 = \sum_{j=1}^m \left(\xi_j \, R_j\right)^2 \, [\cP_{k-2}]_0 \subset [\cP_k]_0$. We know that, for every $P_k \in \cP_k$ and hence also for every $P_k \in [\cP_k]_0$, there exists a unique $H_k \in \cH_k$ and $P_{k-2}Ê\in \cP_{k-2}$ such that $P_k = H_k + \xi^2 \, P_{k-2}$. Applying $\Delta$ to both sides shows that 
$$
\Delta \, P_k = \Delta\,\xi^2 \, P_{k-2}. 
$$
As $\Delta = \sum_{j=1}^m \p_j^2 = \sum_{j=1}^m \left(\p_j \, R_j\right)^2$ is scalar, it maps functions from $[\cP_k]_0$ to $[\cP_{k-2}]_0$. Analogously, $\xi^2 = \sum_{j=1}^m \left( \xi_j \, R_j\right)^2$ maps functions from $[\cP_{k-2}]_0$ to $[\cP_k]_0$. We thus find that 
$$
 \Delta\,\xi^2 \, P_{k-2} = \Delta \, P_k \in [\cP_{k-2}]_0 \quad \Rightarrow \quad P_{k-2}Ê\in [\cP_{k-2}]_0.
$$
Consequently, it must also hold that $H_k \in [\cH_k]_0$. 
 \end{proof}

\begin{cor}
The dimension of $[\cH_k]_0$ is exactly the dimension of the irreducible representation with highest weight $(k,0,\ldots,0)$: 
$$
\dim_\BC \left([\cH_k]_0\right) = \dim_\BC \left([\cP_k]_0\right) - \dim_\BC \left([\cP_{k-2}]_0\right) = \binom{k+m-1}{k} - \binom{k+m-3}{k}. 
$$
\end{cor}
\begin{proof}
As 
$$
[\cP_k]_0 = \textup{span}_{\BC} \left\{ \left( \xi_1\,R_1\right)^{\alpha_1}Ê\ldots \left( \xi_m\,R_m\right)^{\alpha_m}: \ \sum_{j=1}^m \alpha_j = k\right\}, 
$$
we find that the dimension of $[\cP_k]_0$ is exactly $\binom{k+m-1}{k}$. The corollary then follows from the previous lemma. 
\end{proof}

\begin{rem}
Instead of restricting ourselves to this subalgebra $[\cH_k]_0$, one could also work with $2^{2m}$ idempotents    $I$ (see the $L$-representation) and consider maximal left ideals $\cH_k \, I$, analogous as to the monogenic case. 
\end{rem}

\medskip To describe $[\cH_k]_0$ as irreducible $\textup{Spin}(m)$-representation, we consider the isotropic vectors $\f_j, \f_j^\dag \in \BR^1_m$, $j=1,\ldots,n$:
\begin{align*}
\f_j &= \frac{1}{2}Ê\left( R_{2j-1} - i\,R_{2j}Ê\right), 
& \f_j^\dag &= \frac{1}{2}Ê\left( R_{2j-1}Ê+ i\,R_{2j}\right).
\end{align*}
We will now show that the polynomials
$$
f_0^k[1] = \frac{1}{k!}Ê\, \left\langle \xi, \f_1Ê\right\rangle^k[1]
$$
are highest weight vectors for the fundamental representation of $\textup{Spin}(m)$ with weight $(k) = \left(k, 0, \ldots, 0\right)$, i.e. we will show that for the action of the maximal torus $\mathbb{T}$ of $\textup{Spin}(m)$ one has that:
\begin{align*}
H(s) \,f_0^k[1]Ê&= \exp \left( k\, i \, t_1\right) f_0^k[1], \qquad \forall s \in \mathbb{T}.
\end{align*}
Note that, since $2\,\langle \xi, \f_1 \rangle = \xi_1 \,R_1 - i\,\xi_2 \,R_2$, an element $s \in \mathbb{T}$ commutes with $\langle \xi, \f_1 \rangle$; hence, it makes no differences to consider the $H^1$ or $H^0$-action on $f_0^k[1]$. 

\medskip We will start with some auxiliary lemmas. 

\begin{lemma}
For any $k \in \BN$, 
$$
f_0^k[1] = \frac{1}{k!}Ê\left\langle \xi, \f_1Ê\right\rangle^k[1] \in [\cH_k]_0.
$$
\end{lemma}
\begin{proof}
It is immediately clear that $f_0^k[1]$ is homogeneous of degree $k$. We will prove, by induction on $k$, that $\Delta \left\langle \xi, \f_1Ê\right\rangle^k[1] = 0$ or hence $f_0^k[1]Ê\in \cH_k$; since $2\,\left\langle \xi, \f_1Ê\right\rangle = \xi_1 \,R_1 -i\,\xi_2 \,R_2$ it will follow immediately that $f_0^k[1]$ is also in the subalgebra $[\cH_k]_0$.

\medskip As $2\,\left\langle \xi, \f_1Ê\right\rangle = \xi_1 \, R_1 - i\, \xi_2 \, R_2$, we find that $f_0^k[1]$ only contains $\xi_1$ and $\xi_2$ and not $ \xi_3, \ldots, \xi_m$ and thus
\begin{align*}
\left(\p_1 + \p_2\right)^2 \left\langle \xi, \f_1Ê\right\rangle^k[1] &= \frac{1}{2^k}Ê\left( \p_1 + \p_2 \right)^2 \left\langle \xi_1 + \xi_2,  \f_1 \right\rangle^k[1].
\end{align*}
We determine the commutator of $\p_1^2 + \p_2^2$ and $\left\langle \xi, \f_1 \right\rangle$:
\begin{align*}
\left(\p_1^2 + \p_2^2\right) \left\langle \xi, \f_1 \right\rangle &= \left( \p_1^2 + \p_2^2 \right) \left( \xi_1 \,R_1 - i\, \xi_2 \,R_2 \right) \\
&= \left( \xi_1 \,\p_1^2 + 2\,\p_1 \right) R_1 + \xi_1 \, R_1 \, \p_2^2 - i\, \xi_2\,R_2 \, \p_1^2  - i \left( \p_2^2 \,\xi_2 + 2 \,\p_2 \right) \,R_2 \\
&= \left( \xi_1 \,R_1 - i\, \xi_2 \,R_2 \right) \left( \p_1^2 + \p_2^2 \right) + 2\left( \p_1 \, R_1 - i\,\p_2\,R_2 \right) \\
&= \left\langle \xi, \f_1 \right\rangle \left( \p_1^2 + \p_2^2 \right) + 2 \left\langle \p, \f_1\right\rangle.
\end{align*}
Furthermore, let
\begin{align*}
\left\langle \p, \f_1\right\rangle \left\langle \xi, \f_1 \right\rangle &= \left( \p_1 \, R_1 -i \, \p_2 \,R_2 \right) \left( \xi_1 \,R_1 - i\, \xi_2 \,R_2 \right) \\
&= \left( \xi_1 \,\p_1 + 1\right) \, R_1^2 - i\, \xi_2 \,R_2 \,\xi_1 \, \p_1 - i \, \xi_1 \,R_1\, \p_2 \,R_2 - \left( \xi_2 \,\p_2 +1\right) R_2^2 \\
&= \left\langle \xi, \f_1 \right\rangle \left\langle \p, \f_1\right\rangle.
\end{align*}
Let $k =1$, then the degree of homogeneity of $f_0^k[1]$ is one, so it vanishes under the action of $\Delta$. Assume now that $\left( \p_1 + \p_2 \right)^2 \left\langle \xi, \f_1Ê\right\rangle^{k-1}[1] = 0$, then 
\begin{align*}
& \left( \p_1 + \p_2 \right)^2 \left\langle \xi, \f_1Ê\right\rangle^k[1] 
= \left( \p_1 + \p_2Ê\right)^2 \left\langle \xi,\f_1 \right\rangle \left\langle \xi, \f_1Ê\right\rangle^{k-1}[1]Ê\\
&= \left(   \left\langle \xi, \f_1 \right\rangle \left( \p_1^2 + \p_2^2 \right) + 2 \left\langle \p, \f_1\right\rangle \right) \left\langle \xi, \f_1Ê\right\rangle^{k-1}[1]Ê
= 2 \left\langle \xi, \f_1Ê\right\rangle^{k-1}\left\langle \p, \f_1\right\rangle [1]Ê
= 0. 
\end{align*}
\end{proof}

\begin{lemma}
The following commutator relations hold:
\begin{align*}
\left( \xi_1 \pm\, \xi_2 \right) \f_1 &= \f_1^\dag \left( \xi_1 \mp \, \xi_2 \right), &
\left( \xi_1 \pm\, \xi_2 \right) \f_1^\dag &= \f_1 \left( \xi_1 \mp \, \xi_2 \right).
\end{align*}
\end{lemma}
\begin{proof}
We apply the definitions of $\f_1$, $\f_1^\dag$ and use $\left\{ R_j, \xi_k \right\} = 2\,\delta_{jk}Ê\,R_j \, \xi_j$:
\begin{align*}
2 \left( \xi_1 \pm\, \xi_2 \right) \f_1 &= \left( \xi_1 \pm\, \xi_2 \right) \left( R_1 - i \, R_2 \right) = R_1 \, \xi_1 \mp\, R_1 \, \xi_2 + i\, R_2 \, \xi_1 \mp\, i \, R_2 \, \xi_2 \\
&= \left( R_1 + i \, R_2 \right) \left( \xi_1 \mp\, \xi_2 \right) = 2 \, \f_1^\dag \left( \xi_1 \mp\, \xi_2 \right), \\
%
2 \left( \xi_1 \pm\, \xi_2 \right) \f_1^\dag &= \left( \xi_1 \pm\, \xi_2 \right) \left( R_1 + i \, R_2 \right) = R_1 \, \xi_1 \mp\, R_1 \, \xi_2 - i\, R_2 \, \xi_1 \pm\, i \, R_2 \, \xi_2 \\
&= \left( R_1 - i \, R_2 \right) \left( \xi_1 \mp\, \xi_2 \right) = 2 \, \f_1 \left( \xi_1 \mp\, \xi_2 \right).
%
\end{align*}
\end{proof}

\begin{lemma}
For $s \in \mathbb{T}$, let $\eta = \bar{s} \, \xi \, s$, then for $i =1,2$:
\begin{align*}
\eta_i \, R_i &= \left\langle \left( \xi_1 + \xi_2 \right) \exp\left(t_1 \, R_2 \,R_1 \right), R_i \right\rangle.
\end{align*}
And thus
$$
\frac{1}{2} \left(\eta_1 \, R_1 - i\, \eta_2 \, R_2\right) = \left\langle \left( \xi_1 + \xi_2 \right) \exp\left(t_1 \, R_2 \,R_1 \right), \f_1 \right\rangle.
$$
\end{lemma}

\begin{proof}
We first consider $\eta$:
\begin{align*}
\eta &
= \exp\left( -\frac{1}{2}Ê\, \sum_{j=1}^n t_j \, R_{2j}Ê\, R_{2j-1}\right) \sum_{p=1}^n \left( \xi_{2p-1}Ê+ \xi_{2p}Ê\right) \, \exp\left( \frac{1}{2}Ê\, \sum_{j=1}^n t_j \, R_{2j}Ê\, R_{2j-1}\right)  \\
&= \sum_{p=1}^n \left( \xi_{2p-1}Ê+ \xi_{2p}Ê\right) \, \exp\left( t_p \, R_{2p}Ê\, R_{2p-1}\right) \, \exp\left( -\frac{1}{2}Ê\, \sum_{j=1}^n t_j \, R_{2j}Ê\, R_{2j-1}\right) \, \exp\left( \frac{1}{2}Ê\, \sum_{j=1}^n t_j \, R_{2j}Ê\, R_{2j-1}\right)  \\
&= \sum_{p=1}^n \left( \xi_{2p-1}Ê+ \xi_{2p}Ê\right) \, \exp\left( t_p \, R_{2p}Ê\, R_{2p-1}\right).
\end{align*}
Then $\eta_1\,R_1$ is the part of $\eta$ that is commuting with $R_1$, i.e.
\begin{align*}
\eta_1 \, R_1 
%
%
&= \left\langle \left( \xi_1Ê+ \xi_2Ê\right) \, \exp\left( t_1 \, R_2Ê\, R_1 \right), R_1 \right\rangle.
\end{align*}
Similarly, we find that $\eta_2 \, R_2 = \left\langle \left( \xi_1Ê+ \xi_2Ê\right) \, \exp\left( t_1 \, R_2Ê\, R_1 \right), R_2 \right\rangle$. Hence also 
$$
\frac{1}{2}Ê\left(\eta_1 \, R_1 - i\,\eta_2 \, R_2\right) = \left\langle \left( \xi_1Ê+ \xi_2Ê\right) \, \exp\left( t_1 \, R_2Ê\, R_1 \right), \f_1 \right\rangle.
$$
\end{proof}

\begin{thm}
The vectors $f_0^k[1] = \frac{1}{k!}Ê\left\langle \xi, \f_1 \right\rangle^k[1]$ are highest weight vectors of weight $(k) = (k,0,\ldots,0)$, i.e. the action of the maximal torus $\mathbb{T}$ on them is given by 
\begin{align*}
H(s) \,f_0^k[1] &= \exp\left( k\,i \, t_1 \right) f_0^k[1], \qquad \forall s \in \mathbb{T}.
\end{align*}
\end{thm}
\begin{proof}
We will prove the statements by means of induction on $k$. Let $k = 1$ and $s \in \mathbb{T}$, then 
\begin{align*}
H(s) \left\langle \xi, \f_1 \right\rangle[1] &= \left\langle \left( \xi_1 + \xi_2 \right) \exp\left( t_1 \, R_2 \,R_1 \right), \ \f_1 \right\rangle[1] \\
&= \frac{1}{2} \, \f_1 \left( \xi_1 + \xi_2 \right) \exp\left( t_1 \, R_2 \,R_1 \right)[1] + \frac{1}{2} \left( \xi_1 + \xi_2 \right) \exp\left( t_1 \, R_2 \,R_1 \right) \f_1 [1].
\end{align*}
Note that 
\begin{align*}
R_2 \, R_1 \, \f_1[1] &= \frac{1}{2}Ê\, R_2 \, R_1 \left( R_1 - i \,R_2\right)[1] 
= \frac{1}{2} \left( R_2 + i \,R_1 \right)[1] 
= \frac{i}{2} \left( -i\,R_2 + R_1 \right)[1] 
= i \, \f_1[1],\\
R_2 \, R_1 \, \f_1^\dag[1] &= \frac{1}{2}Ê\, R_2 \, R_1 \left( R_1 + i \,R_2\right)[1] 
= \frac{1}{2} \left( R_2 - i \,R_1 \right)[1] 
= \frac{-i}{2} \left( i\,R_2 + R_1 \right)[1] 
= -i \, \f_1^\dag[1]
\end{align*}
and thus 
$$
\exp\left( \pm \, t_1 \, R_2 \, R_1 \right) \f_1[1] = \exp\left( \pm \,i\, t_1 \right) \f_1[1], \qquad \exp\left( \pm \, t_1 \, R_2 \, R_1 \right) \f_1^\dag[1] = \exp\left( \mp \,i\, t_1 \right) \f_1^\dag[1]. 
$$
Now we use $\f_1 \left( \xi_1 + \xi_2\right) = \left(\xi_1 - \xi_2\right) \f^\dag_1$ and $\f^\dag_1 \,\exp\left( t_1 \,R_2\,R_1\right) = \exp\left( -t_1 \,R_2 \,R_1\right) \f^\dag_1$; then we arrive at
\begin{align*}
H(s) \left\langle \xi, \f_1 \right\rangle[1] &= \frac{1}{2} \, \f_1 \left( \xi_1 + \xi_2 \right) \exp\left( t_1 \, R_2 \,R_1 \right)[1] + \frac{1}{2} \, \exp\left( i\,t_1 \right) \left( \xi_1 + \xi_2 \right) \f_1 [1] \\
&= \frac{1}{2} \left( \xi_1 - \xi_2 \right) \exp\left(- t_1 \, R_2 \,R_1 \right)\, \f_1^\dag [1] + \frac{1}{2} \, \exp\left( i\,t_1 \right) \left( \xi_1 + \xi_2 \right) \f_1 [1] \\
&= \frac{1}{2} \,\exp\left(i\, t_1 \right) \left( \xi_1 - \xi_2 \right) \f_1^\dag [1] + \frac{1}{2} \, \exp\left( i\,t_1 \right) \left( \xi_1 + \xi_2 \right) \f_1 [1] \\
&= \frac{1}{2} \, \exp\left(i\, t_1 \right) \left( \f_1 \left( \xi_1 + \xi_2 \right) + \left( \xi_1 + \xi_2 \right) \f_1 \right) [1] \\
&= \exp\left(i\, t_1 \right) \left\langle \xi, \f_1 \right\rangle [1].
\end{align*}

Now take $k >1$ and assume that 
$$
H(s) \left\langle \xi, \f_1 \right\rangle^{k-1}[1] = \exp\left( (k-1)\, i \, t_1 \right) \left\langle \xi, \f_1 \right\rangle^{k-1}[1].
$$
We determine the commutation of $\left\langle \left( \xi_1 + \xi_2 \right) e^{t_1 \,R_2 \,R_1}, \f_1 \right\rangle$ and $\left\langle \xi_1 + \xi_2, \f_1 \right\rangle$:
\begin{align*}
4 & \left\langle \left( \xi_1 + \xi_2 \right) e^{t_1 \,R_2 \,R_1}, \f_1 \right\rangle \left\langle \xi_1 + \xi_2, \f_1 \right\rangle \\
&= \left( \left( \xi_1 + \xi_2 \right) e^{t_1\,R_2\,R_1}Ê\, \f_1 + \f_1 \,\left( \xi_1 + \xi_2 \right) e^{t_1\,R_2\,R_1}\right)  \left( \left( \xi_1 + \xi_2 \right) \f_1 + \f_1 \,\left( \xi_1 + \xi_2 \right)\right) \\
&= \left(\xi_1 + \xi_2 \right) e^{t_1\,R_2\,R_1}Ê\, \f_1 \left( \xi_1 + \xi_2 \right) \f_1 
+ \f_1 \,\left( \xi_1 + \xi_2 \right) e^{t_1\,R_2\,R_1} \, \f_1 \,\left( \xi_1 + \xi_2 \right) \\
&= \left(\xi_1 + \xi_2 \right) \,\f_1 \, e^{-t_1\,R_2\,R_1} \left( \xi_1 + \xi_2 \right) \f_1 
+ \f_1 \,\left( \xi_1 + \xi_2 \right) \f_1 \,e^{-t_1\,R_2\,R_1} \left( \xi_1 + \xi_2 \right) \\
&= \left(\xi_1 + \xi_2 \right) \,\f_1 \left( \xi_1 + \xi_2 \right) e^{t_1\,R_2\,R_1}\, \f_1 
+ \f_1 \,\left( \xi_1 + \xi_2 \right)  \f_1 \,\left( \xi_1 + \xi_2 \right) e^{t_1\,R_2\,R_1} \\
&= \left( \left(\xi_1 + \xi_2 \right) \,\f_1 + \f_1 \,\left( \xi_1 + \xi_2 \right) \right) \left( \left( \xi_1 + \xi_2 \right) e^{t_1\,R_2\,R_1}\, \f_1 + \f_1 \,\left( \xi_1 + \xi_2 \right) e^{t_1\,R_2\,R_1} \right) \\
&= 4 \left\langle \xi , \f_1 \right\rangle \left\langle \left(\xi_1 + \xi_2 \right) \exp\left(t_1 \, R_2 \,R_1 \right), \f_1 \right\rangle.
\end{align*}
Then 
\begin{align*}
H(s) \left\langle \xi, \f_1 \right\rangle^k [1] &= \left\langle \left(\xi_1 + \xi_2 \right) \exp\left(t_1 \, R_2 \,R_1 \right), \f_1 \right\rangle^k [1]Ê\\
&= \left\langle \left(\xi_1 + \xi_2 \right) \exp\left(t_1 \, R_2 \,R_1 \right), \f_1 \right\rangle \left\langle \left(\xi_1 + \xi_2 \right) \exp\left(t_1 \, R_2 \,R_1 \right) , \f_1 \right\rangle^{k-1}[1] \\
&= \left\langle \left(\xi_1 + \xi_2 \right) \exp\left(t_1 \, R_2 \,R_1 \right), \f_1 \right\rangle H(s) \left\langle \xi , \f_1 \right\rangle^{k-1}[1] \\
&= \exp((k-1)\,i\,t_1) \ \left\langle \left(\xi_1 + \xi_2 \right) \exp\left(t_1 \, R_2 \,R_1 \right), \f_1 \right\rangle \left\langle \xi , \f_1 \right\rangle^{k-1}[1] \\
&= \exp((k-1)\,i\,t_1) \  \left\langle \xi , \f_1 \right\rangle^{k-1} \ \left\langle \left(\xi_1 + \xi_2 \right) \exp\left(t_1 \, R_2 \,R_1\right), \f_1 \right\rangle[1] \\
&= \exp((k-1)\,i\,t_1) \ \exp(i \,t_1)\  \left\langle \xi , \f_1 \right\rangle^{k-1} \ \left\langle \left(\xi_1 + \xi_2 \right), \f_1 \right\rangle[1] \\
&= \exp(k\,i\,t_1)\  \left\langle \xi , \f_1 \right\rangle^{k}[1].
\end{align*}
\end{proof}

\begin{rem}
If one considers the orthogonal Spingroup $\textup{Spin}^\perp(m)$ and the analogous representations for $s \in \textup{Spin}^\perp(m)$:
\begin{align*}
H^{0,\perp}(s) f(\xi) &= f(\bar{s}\,\xi\,s), \\
H^{1,\perp}(s) f(\xi) &= s \, f(\bar{s}\,\xi\,s)\,\bar{s}, \\
L^{\perp}(s) f(\xi) &= s\, f(\bar{s}\,\xi\,s),
\end{align*}
the $H^{i,\perp}$, $i=1,2$, resp. $L^\perp$-representations commute with the discrete Laplacian $\Delta$ resp. the discrete Dirac operator $\p$. Let 
$$
g_0^k[1] = \frac{1}{k!}Ê\, \left\langle \xi, \g_1 \right\rangle^k[1], \qquad \g_j = \frac{1}{2}Ê\left( S_{2j-1}\,e_{2j-1}^\perp - i\,S_{2j}\,e_{2j}^\perpÊ\right),
$$
then it can be shown that $g_0^k[1]$ is a highest weight vector of weight $(k)$ with respect to the $H^{0,\perp}$-action. 

\medskip One could then consider simultaneous fundamental representations to both $Spin(m)$-representations. However, the vectors $f_0^k[1]$ are not highest weight vectors for the $\textup{Spin}^\perp(m)$-action. 
To get simultaneous fundamental representations, one would have to consider $\cH_k \, I$, with $I$ the idempotent of the next section and use the highest weight vectors $f_0^k \, I$. 
\end{rem}

\subsection{A primitive idempotent}
Let again $n = \left\lfloor \frac{m}{2}Ê\right\rfloor$ be the truncated half of the dimension $m$. 
%
%
 
\medskip In even dimension, we consider the isotropic basic vectors
\begin{align*}
\f_j &= \frac{1}{2}Ê\left( R_{2j-1} - i\,R_{2j}Ê\right), 
& \f_j^\dag &= \frac{1}{2}Ê\left( R_{2j-1}Ê+ i\,R_{2j}\right), \\
\g_j &= \frac{1}{2}Ê\left( S_{2j-1}\,e_{2j-1}^\perp - i\,S_{2j}\,e_{2j}^\perpÊ\right), 
& \g_j^\dag &= - \frac{1}{2}Ê\left( S_{2j-1}\,e_{2j-1}^\perpÊ+ i\,S_{2j}\, e_{2j}^\perp \right),
\end{align*}
where $j= 1,\ldots,n$. In odd dimension, we consider these basic vectors with the additional basic vectors $R_m$ and $S_m \, e_m^\perp$. 

\begin{lemma}
The basic vectors satisfy
$$
\left\{ \f_j, \f_k^\dag \right\} = \delta_{j,k}, \qquad
\left\{ \g_j, \g_k^\dag \right\} = \delta_{j,k}.
$$
All other commutators are zero. 

\end{lemma}
\begin{proof}
This can immediately be seen by the definition of the basic vectors and the commutator rules 
$$
\left\{ R_j, R_k \right\} = 2\,\delta_{j,k}, \qquad \left\{ S_j\, e_j^\perp, S_k \, e_k^\perp \right\} = - 2\,\delta_{j,k}, \qquad \left\{ R_j, S_k \,e_k^\perp \right\}Ê= 0.
$$
\end{proof}

\begin{lemma}
The Clifford algebra elements $\f_j\, \f_j^\dag$ and $\g_j \,\g_j^\dag$ are all idempotents. Let $I_j = \f_j\, \f_j^\dag\, \g_j \,\g_j^\dag$ then all $I_j$'s are idempotents and $I = \prod_{j=1}^n I_j$ is a primitive idempotent for $m$ even. For $m$ odd one has to add $R_m\,S_m\,e_m^\perp$ to the right of $I$. 
\end{lemma}
\begin{proof}
The fact that they are idempotent follows directly from the commutator rules of the previous lemma. For $m$ even, it is  possible to decompose any element of the Clifford algebra als linear combinations of product of terms $1$, $R_j$, $e_j^\perp\,S_j$ and $R_j \, e_j^\perp \, S_j$, or equivalently as linear combinations of products of the isotropic basis vectors $\f_j$, $\f_j^\dag$, $\g_j$, $\g_j^\dag$:
$$
\BC_{2m} = \textup{Alg}_\BC \left\{ \f_j, \f_j^\dag, \g_j, \g_j^\dag: j =1, \ldots, n\right\}.
$$
As $\f_j \, I = \g_j \, I = 0$, we find that for any $a \in \BC_{2m}$, the element $a \, I$ can be expressed as
$$
a \, I = \prod_{i=1}^n h_i \, I, \qquad h_i \in \left\{ \f_j^\dag, \;
\f_j\,\f_j^\dag,  \; 
\g_j^\dag, \;
\g_j\,\g_j^\dag, \;
\f_j^\dag \,\g_j^\dag, \;
\f_j \, \f_j^\dag \,\g_j^\dag, \;
\f_j^\dag \,\g_j \, \g_j^\dag, \;
\f_j \, \f_j^\dag \,\g_j \, \g_j^\dag\right\}
$$
This is clearly of dimension $8^n = 4^{2n} = 2^{2m}$. A similar count holds for $m$ odd. 
\end{proof}

\section{Fundamental representations with half-integer highest weights}\label{sec:irrep2}
\begin{lemma}
For $k \in \BN$, the vectors $f_0^k\,\I[1]$ are discrete monogenics of degree $k$. 
\end{lemma}
\begin{proof}
We start with the commutator of $\p_1 + \p_2$ and $\langle \xi, \f_1Ê\rangle$:
\begin{align*}
\left[ \left( \p_1 + \p_2 \right), \langle \xi, \f_1 \rangle \right]
&= \frac{1}{2} \left( \p_1 + \p_2 \right) \left( \xi_1 \, R_1 - i\,Ê\xi_2 \, R_2\right) - \frac{1}{2} \left( \xi_1 \, R_1 - i\,Ê\xi_2 \, R_2\right)\left( \p_1 + \p_2 \right) \\
&= \frac{1}{2}Ê\left(R_1 - i\, R_2 \right) = \f_1. 
\end{align*}
It is thus clear that from $\f_1 \, \I = 0$ follows that
$$
 \left( \p_1 + \p_2 \right) \langle \xi, \f_1 \rangle\,\I[1] = \langle \xi, \f_1 \rangle  \left( \p_1 + \p_2 \right)\, \I[1] + \f_1 \, \I[1] = 0.
$$
Now assume that $\left( \p_1 + \p_2\right) \langle \xi, \f_1 \rangle^{k-1}\, \I[1] = 0$, then 
\begin{align*}
\left( \p_1 + \p_2\right) \langle \xi, \f_1 \rangle^k\, \I[1] &= \left(  \langle \xi, \f_1 \rangle  \left( \p_1 + \p_2 \right) + \f_1\right) \langle \xi, \f_1 \rangle^{k-1}\, \I[1] 
= \f_1\, \langle \xi, \f_1 \rangle^{k-1}\, \I[1]. 
\end{align*}
Now note that, from $\f_1^2 = 0$, follows
$$
2\, \f_1 \, \langle \xi, \f_1 \rangle = 2\, \f_1 \, \langle \xi_1 + \xi_2, \f_1 \rangle = 2\, \f_1 \left( \xi_1 + \xi_2 \right) \f_1 = 2\, \langle \xi, \f_1 \rangle \, \f_1. 
$$
It then follows that 
\begin{align*}
\left( \p_1 + \p_2\right) \langle \xi, \f_1 \rangle^k\, \I[1] &= \langle \xi, \f_1 \rangle^{k-1}\, \f_1\, \I[1] = 0. 
\end{align*}
\end{proof}

\begin{thm}
The vector $\left\langle \xi, \f_1 \right\rangle^k\, \I[1]$ are highest weight vectors of weight $\left(k+\frac{1}{2},\frac{1}{2},\ldots,\frac{1}{2} \right)$ under the $L$-action, i.e. the action of the maximal torus $\mathbb{T}$ on them is given by 
$$
L(s) \left\langle \xi, \f_1 \right\rangle^k\, \I[1] = \exp\left( i \left(k+\frac{1}{2}\right) t_1 + \frac{i}{2}\,\sum_{j=2}^n t_j \right) \left\langle \xi, \f_1 \right\rangle^k\, \I[1], \qquad \forall s \in \mathbb{T}.
$$
\end{thm}
\begin{proof}
Let $k \in \BN$, then
\begin{align*}
L(s) \left\langle \xi, \f_1 \right\rangle^k \, \I[1] &= s \left\langle \bar{s}\,\xi\,s, \f_1 \right\rangle^k \, \I[1] = s \,H(s) \left\langle \xi, \f_1 \right\rangle^k\, \I[1] \\
&= \exp\left( k\, i \, t_1 \right) s \left\langle \xi, \f_1 \right\rangle^k\, \I[1].
\end{align*}

Note that $s \left\langle \xi, \f_1 \right\rangle = \left\langle \xi, \f_1 \right\rangle s$ and thus
\begin{align*}
L(s) \left\langle \xi, \f_1 \right\rangle^k \, \I[1] &= \exp\left( k\, i \, t_1 \right) \left\langle \xi, \f_1 \right\rangle^k\, s\, \I[1] = \exp\left( k\, i \, t_1 + \frac{i}{2}Ê\, \sum_{j=1}^n t_j \right) \left\langle \xi, \f_1 \right\rangle^k\, \I[1].
\end{align*}
\end{proof}

\section{Conclusion and future research}
In this paper, we described the spaces of discrete harmonic resp. monogenic polynomials of degree $k$ as representations of a Lie group by constructing a discrete Spingroup which is associated to the linear space of bivectors in appropriate operators. We explicitly made the connection to the special orthogonal Lie algebra representations we found earlier. Using the Spingroup action, we found a more natural way to describe spaces of discrete polynomials as (irreducible) representations than by using the Lie algebra $\mathfrak{so}(m)$. We expect that explicit algorithms for orthonormal basis such as Gel'fand-Tsetlin bases will reduce significantly in complexity compared to using the associated Lie algebra. Furthermore, the definition of the Spingroup gives us a first step in the introduction of simplicial harmonics and simplicial monogenics.

\section*{Acknowledgements}
The first author acknowledges the support of the Research Foundation - Flanders (FWO), grant no. FWO13$\backslash$PDO$\backslash$039.

\end{document}